%% file: nonlinear.tex
\documentclass[journal]{IEEEtran}

\usepackage[english]{babel}

\usepackage{ifpdf}

\usepackage{cite} 
\usepackage{url}
\usepackage{hyperref}

\ifCLASSINFOpdf
	\usepackage[pdftex]{graphicx}
	\graphicspath{{./figures/}}
\else
	\usepackage[dvips]{graphicx}
	\graphicspath{./figures/}
\fi
\usepackage{color}
\usepackage{pgf, tikz, pgfplots}
\usetikzlibrary{shapes, arrows, automata}
\usetikzlibrary{calc,hobby,decorations}

\usepackage[cmex10]{amsmath}
\usepackage{amsfonts, amssymb, amsthm}
\usepackage{mathrsfs}

\usepackage{array}
\usepackage{enumerate}
\usepackage{multirow}
\usepackage{rotating}
\usepackage{subcaption}
\usepackage{siunitx}

\input{my_sections.tex}

\input{mysymbol.sty}

\definecolor{penndarkestblue}{cmyk}{1,0.74,0,0.77}
\definecolor{penndarkerblue}{cmyk}{1,0.74,0,0.70}
\definecolor{pennblue}{cmyk}{0.99,0.66,0,0.57} 
\definecolor{pennlighterblue}{cmyk}{0.98,0.44,0,0.35}
\definecolor{pennlightestblue}{cmyk}{0.38,0.17,0,0.17} 

\definecolor{penndarkestred}{cmyk}{0,1,0.89,0.66}
\definecolor{penndarkerred}{cmyk}{0,1,0.88,0.55}
\definecolor{pennred}{cmyk}{0,1,0.83,0.42} 
\definecolor{pennlighterred}{cmyk}{0,1,0.6,0.24}
\definecolor{pennlightestred}{cmyk}{0,0.43,0.26,0.12} 

\definecolor{penndarkestgreen}{cmyk}{1,0,1,0.68}
\definecolor{penndarkergreen}{cmyk}{1,0,1,0.57}
\definecolor{penngreen}{cmyk}{1,0,1,0.44} 
\definecolor{pennlightergreen}{cmyk}{1,0,1,0.25}
\definecolor{pennlightestgreen}{cmyk}{0.43,0,0.43,0.13}

\definecolor{penndarkestorange}{cmyk}{0,0.65,1,0.49}
\definecolor{penndarkerorange}{cmyk}{0,0.65,1,0.33}
\definecolor{pennorange}{cmyk}{0,0.54,1,0.24} 
\definecolor{pennlighterorange}{cmyk}{0,0.32,1,0.13}
\definecolor{pennlightestorange}{cmyk}{0,0.15,0.46,0.06}
	
\definecolor{penndarkestpurple}{cmyk}{0,1,0.11,0.86}
\definecolor{penndarkerpurple}{cmyk}{0,1,0.13,0.82}
\definecolor{pennpurple}{cmyk}{0,1,0.11,0.71} 
\definecolor{pennlighterpurple}{cmyk}{0,1,0.05,0.46}
\definecolor{pennlightestpurple}{cmyk}{0,0.35,0.02,0.23}
	
\definecolor{pennyellow}{cmyk}{0,0.20,1,0.05} 
\definecolor{pennlightgray1}{cmyk}{0,0,0,0.05}
\definecolor{pennlightgray3}{cmyk}{0.01,0.01,0,0.18}
\definecolor{pennmediumgray1}{cmyk}{0.04,0.03,0,0.31}
\definecolor{pennmediumgray4}{cmyk}{0.08,0.06,0,0.54}
\definecolor{penndarkgray2}{cmyk}{0.09,0.07,0,0.71}
\definecolor{penndarkgray4}{cmyk}{0.1,0.1,0,0.92}

\def\Tr{\mathsf{T}}

\def\bbsigma{\boldsymbol{\sigma}}

\newtheorem{proposition}{\hspace{0pt}\bf Proposition}

\newtheorem{corollary}{\hspace{0pt}\bf Corollary}

\newtheorem{remark}{\hspace{0pt}\bf Remark}

\newtheorem{definition}{\hspace{0pt}\bf Definition}

\renewcommand\blue{\color{black}}

\begin{document}

\title{Invariance-Preserving Localized Activation Functions for Graph Neural Networks}

\author{Luana~Ruiz,~
		Fernando~Gama,~
        Antonio~G.~Marques~
        and~Alejandro~Ribeiro
\thanks{This work in this paper was supported by NSF CCF 1717120, ARO W911NF1710438, ARL DCIST CRA W911NF-17-2-0181, ISTC-WAS and Intel DevCloud, and Spanish MINECO grant TEC2016-75361-R. Preliminary results have been submitted for publication at the ICASSP19 conference \cite{ruiz18-local}. L. Ruiz, F. Gama and A. Ribeiro are with the Dept. of Electrical and Systems Eng., Univ. of Pennsylvania., A. G. Marques is with the Dept. of Signal Theory and Comms., King Juan Carlos Univ.  Email: \{rubruiz,fgama,aribeiro\}@seas.upenn.edu, antonio.garcia.marques@urjc.es.
}
}

\markboth{IEEE TRANSACTIONS ON SIGNAL PROCESSING (ACCEPTED)}%
{Invariance-Preserving Localized Activation Functions for Graph Neural Networks}

\maketitle

\begin{abstract}
Graph signals are signals with an irregular structure that can be described by a graph. Graph neural networks (GNNs) are information processing architectures tailored to these graph signals and made of stacked layers that compose graph convolutional filters with nonlinear activation functions. Graph convolutions endow GNNs with invariance to permutations of the graph nodes' labels. In this paper, we consider the design of trainable nonlinear activation functions that take into consideration the structure of the graph. This is accomplished by using graph median filters and graph max filters, which mimic linear graph convolutions and are shown to retain the permutation invariance of GNNs. We also discuss modifications to the backpropagation algorithm necessary to train local activation functions. The advantages of localized activation function architectures are demonstrated in \blue{four numerical experiments: source localization on synthetic graphs, authorship attribution of 19th century novels, movie recommender systems and scientific article classification.} In all cases, localized activation functions are shown to improve model capacity.
\end{abstract}


\IEEEpeerreviewmaketitle


\section{Introduction} \label{sec:intro}

\input{intro-nonlinear.tex}


\section{Convolutional Processing of Graph Signals} \label{sec:gnn} 

\input{prelim-nonlinear.tex} 


\section{Invariance-Preserving Local Activation Functions} \label{sec:local}

\input{local-nonlinear.tex}

\section{Localized Activation Function Training} \label{sec:bp}

\input{training-nonlinear.tex}


\section{Numerical Experiments} \label{sec:sims}

\input{sims-nonlinear.tex}





\section{Conclusions} \label{sec:conclusions}

\input{conclusions-nonlinear.tex} 


\appendices


\bibliographystyle{IEEEtran}
\bibliography{myIEEEabrv,bib-nonlinear}

\end{document}

%% file: my_sections.tex
\usepackage{needspace}

\newcommand{\myparagraph}[1]{\needspace{1\baselineskip}\medskip\noindent {\bf #1}}



%% file: intro-nonlinear.tex


With a local structure that repeats itself at every point, images and time signals are characterized by their regularity. However, many problems in contemporary information processing ---such as analyzing texts or designing recommender systems--- leverage data with support on irregular structures. These types of data are best represented as graph signals, which explains the growing interest in devising architectures capable of exploiting the structural information carried by the graph topology. In particular, graph convolutional neural networks (GNNs) \blue{\cite{bruna14-deepspectralnetworks, henaff2015deep, defferrard17-cnngraphs, kipf17-classifgcnn, ying18-pinterest, gama18-gnnarchit, xu2018how}} have been at the center of attention in an effort to reproduce the remarkable success that convolutional neural networks (CNNs) achieved in processing images and time signals \cite{kuo17-recos}.

CNNs are made out of stacked layers that each comprise two basic operations: a bank of {\it trainable} linear convolutional filters and a {\it fixed} nonlinear activation function. GNNs retain this basic architecture but replace standard convolutions with graph convolutional filters \cite{gama18-gnnarchit}. Graph convolutional filters are built by reinterpreting a linear graph diffusion operator as a shift. Following this interpretation, a convolution is simply defined as the sum of scaled shift compositions. This allows us to further define linear transforms that are polynomials on the shift (diffusion) operator and which were seen to be proper generalizations of linear time invariant filters \cite{segarra17-linear}. Regarding activation functions, GNNs utilize the same functions as CNNs ---rectified linear units (ReLUs), sigmoids, and hyperbolic tangents---, and, like in CNNs, these functions are either applied locally \cite{defferrard17-cnngraphs, kipf17-classifgcnn} or within local neighborhoods when mixed with nonlinear pooling operations \blue{\cite{gama18-gnnarchit, henaff2015deep, bruna14-deepspectralnetworks}}.

Since the notion of a neighborhood in an image or a time signal is always the same, there is no reason to adapt the activation function for different tasks. However, the structure of a neighborhood may vary significantly from graph to graph. This motivates the design of activation functions that are {\it adapted} to the graph structure. In this paper, we propose not only to adapt the nonlinear activation function to the structure of the graph, but also to make it \blue{{\it multiresolution} and {\it trainable} by assigning linear weights to the value of the nonlinear function at $1, 2, \ldots, K$-hop node neighborhoods.} In particular, we introduce two types of local activation functions based on median and maximum graph filters \cite{segarra16-globalsip, segarra17-camsap}. Median and maximum graph filters are the graph signal processing (GSP) counterparts of the rank filters studied in traditional signal processing \cite{hodgson85-rank}. Similarly to linear graph filters, median and maximum filters encode graph structural information, but they do so through an implicit, nonlinear dependence on node neighborhoods of the graph. \blue{As such, they can be used to design multiresolution activation functions that learn to assign different weights to different neighborhood resolutions (Section \ref{sec:local}).} 

A fundamental consideration in the introduction of trainable activation functions is retaining the permutation invariance of GNNs \blue{\cite{ZouLerman18-Scattering, gama19-stability, gama19-neurips}}. \blue{Observe that this form of invariance is with respect to a reordering of the nodes' labels that retains the structure of the graph. This is different from the definition of permutation invariance used in set theory, which requires invariance with respect to all possible reorderings \cite{rezatofighi17-deepsetnet, zaheer17-deepsets}.} Indeed, since graph convolutional filters are polynomials on the diffusion operator, they can be readily shown to be invariant to permutations or node relabelings (Proposition \ref{prop_invariance}). This is an important property because it renders the processing of graph signals independent of the choice of node labels and is an effective way of exploiting the internal symmetries of graph signals. Pointwise activation functions are independent of the graph structure and, as such, do not affect permutation invariance. We will show here that nonlinear activation functions based on median and maximum filters, although local, are still invariant to permutations (Propositions \ref{nl_prop_invariance} and \ref{max_nl_prop_invariance}).

It is important to mention that the use of trainable activation functions in CNNs \cite{chandra04-sigmoid, he15-rect, goodfellow13-maxout, goodfellow16-deeplearn} and GNNs \cite{scardapane18-kernelactiv} has been pursued with success but that the existing literature focuses on learning pointwise activations functions. In particular, \cite{chandra04-sigmoid} uses a sigmoid function that is exponentiated by a trainable parameter to change the activation function slope in feedforward neural networks. In \cite{he15-rect}, the ReLU activation is paired with a regularizer whose importance is controlled by a trainable linear weight variable. Perhaps the best known example is that of maxout units \cite{goodfellow13-maxout, goodfellow16-deeplearn}, which replace pointwise max functions by the maximum among a number of filters at a given layer. In the specific case of activation functions for GNNs, the work in \cite{scardapane18-kernelactiv} proposes learning a pointwise activation function parametrized by a dictionary. All of these papers differ from the goal of our paper, which is to design activation functions that learn to assign different weights to different neighborhood resolutions. 

This paper is organized as follows. Section \ref{sec:gnn} describes a basic GNN architecture and lays the ground for the localized activation functions introduced in Section \ref{sec:local}. In Sections \ref{sbs:neigh_med} and \ref{sbs:neigh_max} respectively, we define a class of linearly parametrized multiresolution median and max graph filters, which are further interpreted as permutation invariant localized activation functions whose linear parameters can be trained to learn appropriate localized activation functions for GNNs. In Section \ref{sbs:backpropagation}, we address backpropagation training of the multi-hop median and multi-hop max operators. We conclude Section \ref{sec:bp} with considerations on the additional computational complexity incurred by the localized activation functions we propose (Section \ref{sbs:complexity}). In Section \ref{sec:sims}, the performance of our localized activation functions is evaluated on both synthetic and real-world datasets, including the problems of authorship attribution, \blue{movie recommendation systems and scientific article classification}. Concluding remarks are presented in Section \ref{sec:conclusions}.

\myparagraph{Notation.} Bold uppercase $\bbA$ is a matrix and bold lowercase $\bbx$ is a vector. We use $[\bbA]_{ij}$ to denote the $(i,j)$ entry of $\bbA$ and $[\bbx]_i$ for the $i$th entry of $\bbx$. Transposition of vectors and matrices are written as $\bbx^{\Tr}$ and $\bbA^\Tr$, respectively. The operation $\diag(\bbA)$ is a column vector with entries $[\diag(\bbA)]_i = [\bbA]_{ii}$. The operation $\diag(\bbx)$ yields a diagonal matrix with $[\diag(\bbx)]_{ii} = [\bbx]_{i}$. Calligraphic uppercase $\ccalT$ is a set with cardinality $|\ccalT|$.

%% file: prelim-nonlinear.tex

%
We consider graph signals $\bbx = [x_1, \ldots, x_N]^\Tr \in \reals^N$ in which the component $x_i$ is associated with the $i$th node of a weighted and directed graph $\ccalG$. The graph is composed of a vertex set $\ccalV = \{1,\dots,N\}$, an edge set $\ccalE \subseteq \ccalV \times \ccalV$ of ordered pairs $(i,j)$ and a weight function $\ccalA: \ccalE \to \mbR$ taking values $\ccalA(i,j) = a_{ij}$. The presence of the edge $(i,j)$ in the set $\ccalE$ is interpreted as an expectation that signal components $i$ and $j$ are close. The weight $a_{ij}$ measures the expected similarity between nodes, i.e., the larger $a_{ij}$, the more related components $x_i$ and $x_j$ should be. Associated with the graph $\ccalG$, we also define the weighted adjacency matrix $\bbA \in \reals^{N \times N}$ and the unweighted adjacency matrix $\bbN \in \{0,1\}^{N \times N}$. Both of these matrices have a sparsity pattern that matches that of the edge set of the graph by taking values $[\bbA]_{ij}=[\bbN]_{ij}=0$ for all $(j,i) \notin \ccalE$. At entries corresponding to edges of $\ccalG$, we have
\begin{equation} \label{eqn:adjacency_matrices}
   [\bbA]_{ij} = a_{ji}, \quad
   [\bbN]_{ij} = 1,      \quad 
   \text{for all\ } (j,i) \in \ccalE\ .
\end{equation}
A graph shift operator associated with $\ccalG$ is a sparse matrix $\bbS$ whose nonzero entries are either in the diagonal or at entries that match an edge of the graph,
\begin{equation} \label{eqn:gso_sparsity_pattern}
   [\bbS]_{ij} = s_{ij} \neq 0, \quad
   \text{if\ }  i=j, \
   \text{or\ } (j,i) \in \ccalE\ .
\end{equation}
The nonzero entry pattern in \eqref{eqn:gso_sparsity_pattern} is meant to abstract the properties that are shared between different matrix representations. The weighted adjacency $\bbA$ and the unweighted adjacency $\bbN$ defined in \eqref{eqn:adjacency_matrices} satisfy the restrictions placed by \eqref{eqn:gso_sparsity_pattern} on the shift operator $\bbS$. Other acceptable choices are the corresponding Laplacians $\bbL_\bbA := \diag(\bbA \bbone) - \bbA$ and $\bbL_\bbN := \diag(\bbN \bbone) - \bbN$ as well as the self loop adjacency $\bbM = \bbI + \bbN$. Using $\bbS$ as a stand-in for an arbitrary matrix representation of the graph avoids restricting attention to a particular selection. This is useful because while different matrix representations are of interest in different contexts, they can all be leveraged in a similar manner to process graph signals $\bbx$ using graph convolutions as we explain in Section \ref{sec_graph_convolutions}.

The graph shift operator induces neighborhoods in the graph. The 1-hop neighborhood of node $i$ is the set of nodes $\ccalN_i = \neighborhood{i}{1} := \{ j :(j,i) \in \ccalE \}$ that can be reached from $i$ by taking a single hop along an edge $(j,i)$. More generically, the $k$-hop neighborhood is the set of nodes that can be reached in exactly $k$ hops. This set is easily determined from the nonzero elements of the $k$th power of the unweighted adjacency matrix $\bbS = \bbN$,
\begin{equation} \label{eqn:neigh}
   \neighborhood{i}{k} := 
        \left\{j \ :\ \big[\bbS^k\big]_{ij} \neq 0\right\}\ .
\end{equation}
Observe that, consistent with \eqref{eqn:neigh}, the $0$-hop neighborhood is the node itself since for $k=0$ we have $\bbS^k=\bbS^0=\bbI$. In general, $\ccalN_{i}^{l} \nsubseteq \ccalN_{i}^{k}$ for $l < k$ since a node may be reachable in exactly $k$ hops but not reachable in exactly $l < k$ hops. A sufficient condition for having 
$\ccalN_{i}^{l} \subseteq \ccalN_{i}^{k}$ for $l < k$ is that the shift operator be nonnegative with a full diagonal. An example of a shift operator with this property is the self loop unweighted adjacency matrix $\bbM=\bbI+\bbN$; see Remark \ref{rmk_choice_of_shift}. Further note that we can think of $\bbS^k$ itself as the shift operator of a graph. With this interpretation, the $k$-hop neighborhood of $\bbS$ is equivalent to the 1-hop neighborhood of $\bbS^k$. We will exploit this fact to simplify definitions in Section \ref{sec:local}.

%
\subsection{Graph Convolutions}\label{sec_graph_convolutions}

A graph convolution is defined as a linear operator $\bbH$ that can be written as a polynomial in the shift operator operator $\bbS$ \blue{\cite{sandryhaila13-dspg, sandryhaila14-freq, shuman13-mag, du2018graph}}. Formally, for a given vector of coefficients $\bbh=[h_0, \ldots, h_{K-1}] \in \reals^K$ and a graph signal $\bbx$, the graph convolution of $\bbh$ and $\bbx$ is an operation whose outcome is the graph signal
\begin{equation} \label{eqn:lsigf}
   \bbz 
        \  = \ \sum_{k=0}^{K-1} h_k \bbS^k\bbx
        \ := \ \bbh  *_\bbS \bbx
\end{equation}
where we have defined the graph convolution operator $*_\bbS$ to represent the linear transformation in \eqref{eqn:lsigf}. The graph convolution in \eqref{eqn:lsigf} shares the localization properties of regular convolutions since each of the terms in the polynomial performs operations that are localized to a specific neighborhood. Indeed, it is straightforward to see that we can have $[\bbS^k]_{ij}\neq 0$ only when $j$ is in the $k$-hop neighborhood of $i$. Consequently, the $i$th entry of the product $\bbS^k\bbx$ is only affected by the entries $x_j$ for which $j\in\neighborhood{i}{k}$. We can then think of the first polynomial term $h_0 \bbS^0\bbx$ as a nodewise operation, the second polynomial term  $h_1 \bbS^1\bbx$ as a 1-hop neighborhood operation, and, in general, the $k$th polynomial term as an operation localized to $(k-1)$-hop neighborhoods. This is akin to regular convolutional filters of order $K$ extending to no more than $K-1$ points in time. This property makes the graph convolution in \eqref{eqn:lsigf} a natural choice for the extension of convolutional neural networks to signals supported on graphs, as we discuss in the following section. 

%
\begin{remark}\normalfont We point out that \eqref{eqn:lsigf} is also referred to as a linear shift invariant (LSI) filter because of its invariance with respect to the application of the shift operator. Namely, if the input $\bbx$ is replaced by the (shifted) input $\bbS\bbx$ the output shifts from $\bbz$ to \bbS\bbz. This is analogous to (convolutional) linear time invariant (LTI) filters in which a time shift of the input produces a time shift at the output. LSI filters also admit spectral representations in terms of graph Fourier transforms that are analogous to spectral representations of time invariant filters \cite{sandryhaila13-dspg}. The connections between LSI and LTI filters are deeper than simple analogies. Regular convolutions and regular Fourier transforms can be recovered if the graph shift operator is particularized to a cycle graph representing a periodic time axis.
\end{remark}

%
\subsection{Graph Neural Networks}

Consider a training set $\ccalT = \{(\bbx_m, \bby_m)\}$ comprised of $|\ccalT|$ input-output samples $(\bbx_m, \bby_m)$. In each training example, the vector $\bbx_{m}$ is a graph signal and the vector $\bby_{m}$ is some observed output whose shape depends on the problem at hand. For instance, $\bby_{m}$ might be a class label in a classification problem, or another graph signal in the context of regression. The objective of learning is to find a representation of the training set that can produce output estimates $\hby (\bbx)$ for unknown inputs $\bbx \notin \ccalT$. Graph Neural Networks (GNNs) do so by composing computational layers that are themselves the composition of two distinct operations, namely: (i) a filter bank of graph convolutions, and (ii) a pointwise nonlinear activation function \cite{kipf17-classifgcnn}. Formally, each layer $\ell$ takes at its input a set of $F_{\ell-1}$ features $\bbx_{\ell-1}^f$, where the $\bbx_{\ell-1}^f$ are signals supported on $\ccalG$ to which we attach the graph shift operator $\bbS$. Each  feature is processed by a separate set of filter banks, and each filter bank consists of $F_{\ell}$ graph convolutional filters as described by equation \eqref{eqn:lsigf}. Denoting the filters' coefficients by $\bbh_\ell^{fg} \in \reals^K$, we can write the resulting $F_{\ell}$ convolutional features $\bbu_\ell^{f}$ as
\begin{equation} \label{eqn:gnn_linear_summary}
   \bbu_{\ell}^{f} 
       \ =\ \sum_{g=1}^{F_{\ell-1}} \left(\bbh^{f g}_{\ell}  *_\bbS \bbx \right)\ 
       \ =\ \sum_{g=1}^{F_{\ell-1}} \sum_{k=0}^{K-1} h^{f g}_{\ell k} \ \bbS^k \bbx_{\ell-1}^{g}\ .
\end{equation}

The convolutional features $\bbu_{\ell}^{f}$ are then fed into a scalar and nonlinear activation function to produce the $\ell$th layer's output features,
\begin{equation} \label{eqn:pointwise_nonlinearity}
   \bbx_{\ell}^{f} = \sigma \left(\bbu_{\ell}^{f}\right)
\end{equation}
which will then act as the subsequent layer's input features.

Beginning with $\bbx_0=\bbx$ as the first layer's ($\ell=1$) input, we proceed recursively through \eqref{eqn:gnn_linear_summary}-\eqref{eqn:pointwise_nonlinearity} until the last layer. In a system with $L$ layers, the GNN output is the collection of features $\bbx_L = [(\bbx_{L}^{1})^{\Tr},\ldots,(\bbx_{L}^{F_{L}})^{\Tr}]^{\Tr}$, and represents the output estimate $\hby(\bbx)=\bbx_L$. For future reference, we define the set $\ccalH = \{\bbh_{\ell}^{fg}\}_{\ell,f,g}$ grouping all the graph filter coefficients of the model to write the GNN as a mapping $\Phi : \bbx_0 \mapsto \bbx_L$ parametrized by $\bbS$ and by the coefficient set $\ccalH$:
\begin{equation} \label{eqn:GNN_transform}
   \Phi(\bbx; \bbS, \ccalH) 
       \ = \ \bbx_L 
       \ = \ \hby(\bbx)\ .
\end{equation}
In \eqref{eqn:GNN_transform}, $\bbS$ is given, while $\ccalH$ is determined by optimizing a loss function $\bbarJ = \sum_{\ccalT} J[\bby,\hby(\bbx_{m})]=\sum_{\ccalT} J[\bby,\Phi(\bbx_{m}; \bbS, \ccalH)]$ on the training set $\ccalT$.

GNNs are particular cases of neural networks (NNs) and generalizations of convolutional neural networks (CNNs). NNs are obtained through the composition of arbitrary linear operations with pointwise activation functions. In practice, they are hard to train and underperform architectures exploiting structural information because the number of parameters to learn is too large. CNNs resolve this problem for time signals, images, and other signals supported on regular domains by restricting arbitrary linear transformations to convolutional filter banks \cite{lecun15-deeplearning}. In graph domains, GNNs fulfill the same purpose by utilizing graph convolutions [cf. \eqref{eqn:lsigf} and \eqref{eqn:gnn_linear_summary}]. 


The pointwise nonlinear operators typically used in GNNs, however, neglect the graph structure. They process nodes homogeneously thus ignoring the different compositions of their heterogeneous neighborhoods. This is not unreasonable for the regular signals that are processed with CNNs because their underlying regular structures repeat themselves from node to node and layer to layer. On a generic irregular graph, however, the individual node connectivity profiles are important because they explain how each node interacts with the other nodes in the graph. In this context, summarizing nonlinear operations acting on node neighborhoods instead of individual nodes would exploit meaningful information that the mere application of a graph filter followed by a pointwise nonlinearity do not. Preliminary evidence for this ability comes from the better reconstruction properties of median (nonlinear) graph filters and the improvement in topology identification stemming from the use of nonlinear structural models \cite{shen18-online}. The goal of this paper is to design nonlinear local activation functions for GNNs that leverage the neighborhood structure of the graph.

%% file: local-nonlinear.tex


In pursuing nonlinear activation functions that exploit the graph structure it is important to begin by understanding why graph convolutions are suitable for processing graph signals. A partial answer to this question is the fact that processing signals with a GNN is independent of the graph labeling, as we formally state and prove in the following proposition. 

%
\begin{proposition}\label{prop_invariance}
Consider a graph signal $\bbx$ supported on a graph $\ccalG$ with shift operator $\bbS$. Let $\Phi(\bbx; \bbS, \ccalH)$ be the output of a graph neural network (GNN) with coefficient set $\ccalH$ [cf. \eqref{eqn:gnn_linear_summary}-\eqref{eqn:GNN_transform}]. If $\bbP$ is a permutation matrix, then
\begin{equation} \label{eqn_prop_invariance}
   \Phi\Big(\bbP^\Tr\bbx; \bbP^\Tr\bbS\bbP, \ccalH \Big) 
      = \bbP^\Tr\Phi\Big(\bbx; \bbS, \ccalH\Big)
\end{equation}
i.e., the output of the GNN is invariant to permutations.
\end{proposition}
%
%
\begin{proof}
The result follows from the fact that the LSI filter in \eqref{eqn:lsigf} is invariant to permutations. To see this, consider the graph permutation $\bbS' = \bbP^\Tr\bbS\bbP$ and let $\bbx' = \bbP^\Tr \bbx$ be a signal on the permuted graph, with $\bbx$ its non-permuted counterpart. The application of \eqref{eqn:lsigf} to $\bbx'$ yields
\begin{equation*}
   \bbz' \  = \ \sum_{k=0}^{K-1} h_k {\bbS'}^k  \bbx'
         \  = \ \sum_{k=0}^{K-1} h_k (\bbP^\Tr\bbS\bbP)^k \bbP^\Tr\bbx\ .
\end{equation*}
Using the fact that $\bbP^k = \bbP$ and that permutation matrices are orthogonal, we get
%
%
%
\begin{equation*} 
   \bbz' \  = \ \bbP^\Tr \bigg(\sum_{k=0}^{K-1} h_k \bbS^k\bbx\bigg)
         \  = \ \bbP^\Tr \bbz
\end{equation*}
which is simply the permutation by $\bbP$ of the output of filter \eqref{eqn:lsigf} applied to $\bbx$. As for activation functions, the fact that they are scalar automatically entails invariance to permutations. Thus, if the input to a GNN layer is permuted, its output will be permuted likewise. This is also true in architectures with multiple layers, where permutations will cascade from one layer to the other until they reach the output. \end{proof}

%
\begin{figure*}[t]
	\centering
	\begin{subfigure}{0.25\textwidth}
		\centering
		\includegraphics[width=0.9\textwidth]{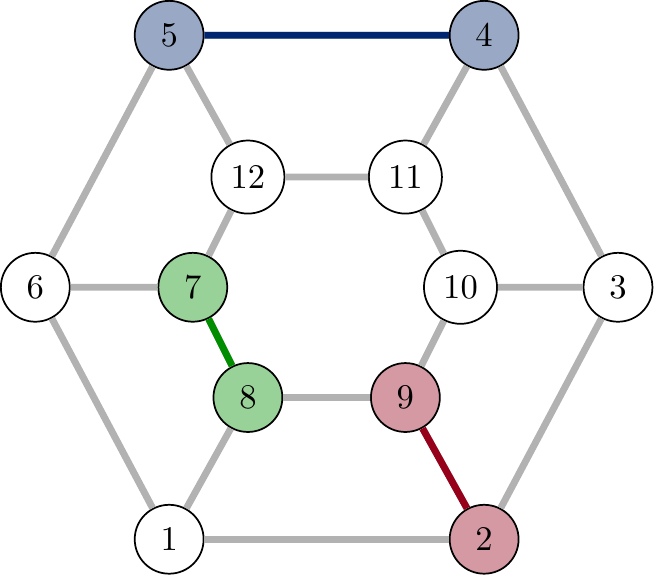}
		\caption{Graph $\ccalG$ and signal $\bbx$}
		\label{original}
	\end{subfigure}
		\hfill
	\begin{subfigure}{0.25\textwidth}
		\centering
		\includegraphics[width=0.9\textwidth]{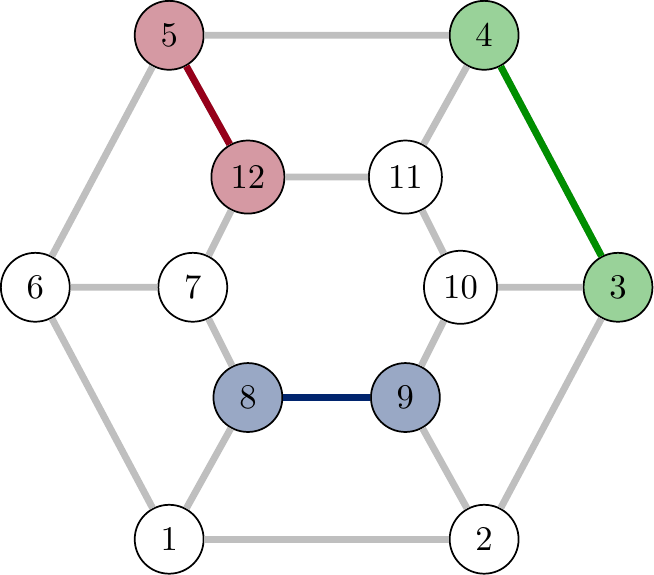} 
		\caption{Graph $\ccalG$ and permuted signal $\bbx'$}
		\label{permuted}
	\end{subfigure}
	\hfill
	\begin{subfigure}{0.25\textwidth}
		\centering
		\includegraphics[width=0.9\textwidth]{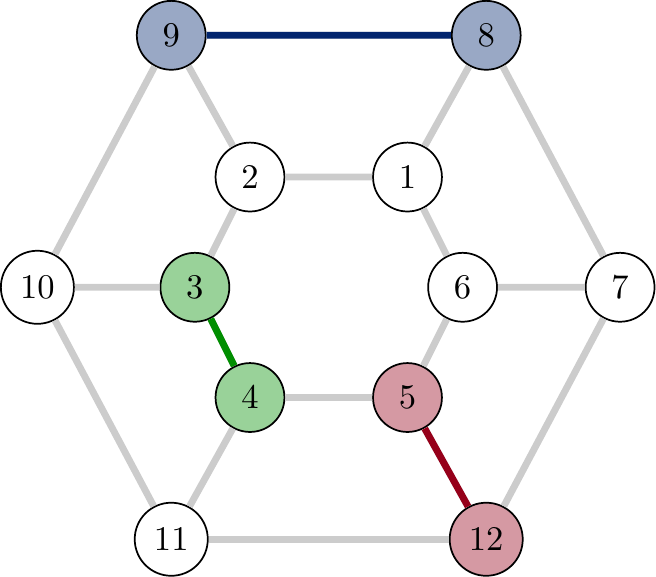} 
		\caption{Permuted graph $\ccalG'$ and permuted signal $\bbx'$}
		\label{permuted}
	\end{subfigure}
\caption{Permutation invariance of graph neural networks (GNNs). The output of a GNN is invariant to graph permutations (Proposition \ref{prop_invariance}). This not only means independence from labeling but it also shows that GNNs exploit internal signal symmetries. The signals on (a) and (b) are different signals on the same graph but they are permutations of each other -- interchange inner and outer hexagons and rotate $180^\circ$ (c.f. (c)). A GNN would learn how to classify the signal in (b) from seeing examples of the signal in (a). (Integer node labels ranging from 1 to 12. Signal values are represented by different colors.)} 
	\label{fig:permutation_invariance}
\end{figure*}	


The permutation invariance stated in Proposition \ref{prop_invariance} shows that the features that are learned by a GNN are independent of the labeling of the graph. But permutation invariance is also important because it means that GNNs exploit internal signal symmetries as we illustrate in Figure \ref{fig:permutation_invariance}. The graphs in Figure \ref{fig:permutation_invariance}-(a) and Figure \ref{fig:permutation_invariance}-(b) are the same, as indicated by the integer labels. The signals in Figure \ref{fig:permutation_invariance}-(a) and Figure \ref{fig:permutation_invariance}-(b) are different, as indicated by different colors. However, it is possible to permute the graph onto itself to make the signals match -- rotate $180^\circ$ degrees and pull it inside out (Figure \ref{fig:permutation_invariance}-(c)). It then follows from Proposition \ref{prop_invariance} that the output of a GNN applied to the signal on the left (a) is a corresponding permutation of the output of the same GNN applied to the signal on the right (b). This is beneficial because we can learn to process the signal on (a) from seeing examples of the signal on (b). Although most  graphs do not have perfect symmetries, existing perturbation analyses show that similar comments hold for transformations that are close to perturbations \cite{gama19-stability}. In the design of nonlinear, localized activation functions, our goal is to retain this property by making localized activation functions permutation invariant. 


%

\subsection{Median GNNs} \label{sbs:neigh_med} 

Consider a set $\ccalX=\{x_1,\ldots, x_n\}$ and order the elements of $\ccalX$ so that $x_{[1]} \leq x_{[2]} \leq \ldots \leq x_{[n]}$. If we denote by $n \div 2$ the integer division of the cardinality of $\ccalX$ by 2, the median of the set $\ccalX$ is given by
\begin{equation} \label{eqn_set_median}
   \text{med}(\ccalX) = x_{[n \div 2 + 1]}\ .
\end{equation}
In order to define activation functions in terms of median graph filters, we begin by defining the median operator associated with a graph shift operator $\bbS$. 

%
\begin{definition}[Median operator] \label{def_median_operator}
Let $\bbS$ be a graph shift operator and let $\ccalN_i$ denote the neighborhoods induced by $\bbS$, each $i$, $1 \leq i \leq N$, being a node of the graph. The output of the median operator $\text{med} (\bbS, \cdot)$ applied to the graph signal $\bbx$ is the graph signal $\bbz  := \text{med} (\bbS, \bbx)$, whose components we write
\begin{equation} \label{eqn_def_median_operator}
   \big[\bbz\big]_i
         = \big[ \text{med} (\bbS, \bbx)\big]_i
         =       \text{med} \big( \{ x_j : j\in\ccalN_i \}\big)\ .
\end{equation} \end{definition}

%
Definition \ref{def_median_operator} implies that the median operator replaces the value of the analyzed signal $\bbx$ at each node by the median of the values of $\bbx$ in the corresponding neighborhood. In that sense, we can think of $\text{med} (\bbS, \cdot)$ as a nonlinear diffusion operator in which, instead of computing the average of neighboring values at each node, we compute their median. These two diffusions tend to yield similar values if neighboring values are symmetric around their mean. 

From the definition of the median operator, we can now define what we call the multiresolution median graph filter as a linear combination of medians associated with neighborhoods of different depths. We write it as follows.

%
\begin{definition}[Multiresolution median graph filter] \label{def_median_filter}
Given a graph shift operator $\bbS$ and a vector of $K+1$ filter coefficients $\bbw=[w_0, \ldots, w_K]^\Tr$, the output of the median graph filter with coefficients $\bbw$ applied to the signal $\bbx$ is the signal $\bbz$,
\begin{equation} \label{eqn_def_median_filter}
   \bbz  := \sum_{k=0}^{K} w_k\, \text{med} (\bbS^k, \bbx)\ .
\end{equation} \end{definition}

%
This definition is to be contrasted with that of the linear graph convolution in \eqref{eqn:lsigf}. In \eqref{eqn:lsigf}, each summand $\bbS^k\bbx$ represents a weighted average, so that the value of the signal at node $i$ is affected by the values of the signals at its $k$-hop neighbors. In \eqref{eqn_def_median_filter}, the summand $\text{med} (\bbS^k, \bbx)$ summarizes information from the same set, because the shift operator $\bbS^k$ is associated with a graph in which nodes are connected to their $k$-hop neighbors. Differently from regular linear filters, however, the signal values at neighboring nodes are not averaged, but summarized by the median operation. We can thus think of \eqref{eqn_def_median_filter} as a nonlinear convolutional filter constructed from median operations \blue{that retains the multiresolution aspect of conventional linear graph convolutions.} 

The median graph filter in Definition \ref{def_median_filter} is used here to define nonlinear activation functions. Specifically, consider a collection of median filter coefficients $\bbw_{\ell}^f$ and replace the pointwise activation functions in \eqref{eqn:pointwise_nonlinearity} by the \textit{local activation functions}
\begin{equation} \label{eqn_median_activation}
   \bbx_{\ell}^{f} 
       = \sum_{k=0}^{K}w_{\ell k}^f\,\text{med}(\bbS^k,\bbu_{\ell}^f)\ .
\end{equation}

A median GNN is a GNN with median activation functions, i.e., one in which the $\ell$th layer is defined by the composition of \eqref{eqn:gnn_linear_summary} and \eqref{eqn_median_activation}. We can interpret this architecture as the composition of two types of convolutional layers: linear convolutional layers [cf. \eqref{eqn:gnn_linear_summary}] and nonlinear convolutional layers based on median filters \eqref{eqn_median_activation}. Observe that the output of a median GNN is determined by the linear filter coefficients $\ccalH$ and the median graph filter coefficients $\ccalW = \{\bbw_{\ell}^{f}\}_{\ell,f}$. For future reference, we therefore define the median GNN map
\begin{equation} \label{eqn_median_GNN_transform}
   \Phi(\bbx; \bbS, \ccalH, \ccalW) 
       \ = \ \bbx_L 
       \ = \ \hby(\bbx)
\end{equation}
where $\bbx_L$ is obtained from the application of $L$ median GNN layers, each of which computes a linear convolution \eqref{eqn:gnn_linear_summary} followed by a median graph filter activation \eqref{eqn_median_activation}. That the definition of median graph filters parallels the definitions of graph convolutions is not an accidental choice. Choosing activation functions of this form is intended \blue{not only to encode the graph structure at multiple resolutions, but also} to preserve the permutation invariance of GNNs, as we formally state and prove next.

%
\begin{proposition}\label{nl_prop_invariance}
Consider a graph signal $\bbx$ supported on a graph with shift operator $\bbS$. Let $\Phi(\bbx; \bbS, \ccalH, \ccalW)$ be the output of a median GNN with linear filter coefficients $\ccalH$ and nonlinearity coefficients $\ccalW$ [cf. \eqref{eqn:gnn_linear_summary}-\eqref{eqn:GNN_transform}, \eqref{eqn_median_activation}. If $\bbP$ is a permutation matrix, then
\begin{equation} \label{eqn_nl_prop_invariance}
   \Phi\Big(\bbP^\Tr\bbx; \bbP^\Tr\bbS\bbP, \ccalH, \ccalW \Big) 
      = \bbP^\Tr\Phi\Big(\bbx; \bbS, \ccalH, \ccalW \Big)
\end{equation}
i.e., the output of the median GNN is invariant to permutations.
\end{proposition}
%
%
\begin{proof}
The result follows from Proposition \ref{prop_invariance} and the fact that the median activation function in \eqref{eqn_median_activation} is invariant to permutations. Consider the graph permutation $\bbS' = \bbP^\Tr\bbS\bbP$ and let $\bbx' = \bbP^\Tr \bbx$ be the corresponding permuted graph signal. Applying the median activation function to $\bbx'$, we get
\begin{equation*}
   \bbz' \  = \ \sum_{k=0}^{K} \omega_k \text{med} (\bbS'^k, \bbx') = 
         \ \sum_{k=0}^{K} \omega_k \text{med} ((\bbP^\Tr)^k \bbS^k \bbP^k, \bbP^\Tr \bbx)
\end{equation*}
but this expression can be simplified by observing that $\bbP^k = \bbP$, $\bbP^\Tr = \bbP^{-1}$ and that the graph median operators in equations \eqref{eqn_def_median_operator} and \eqref{eqn_def_median_filter} is permutation invariant. This yields
%
%
%
\begin{equation*} 
   \bbz' \  = \ \bbP^\Tr \bigg(\sum_{k=0}^{K} \omega_k \text{med} (\bbS^k, \bbx)\bigg)
         \  = \ \bbP^\Tr \bbz
\end{equation*}
a result that will also hold for multi-layered median GNNs as graph permutations cascade from each layer to the next.
\end{proof}

As nonscalar operators that take in values of a graph signal in localized neighborhoods around each node, median activation functions endow GNNs with the ability to extract nonlinear local features in addition to the linear features extracted through graph convolutions. This is something that pointwise activation functions, by construction, cannot do. Proposition \ref{nl_prop_invariance} establishes that the features extracted by median graph filters are invariant to permutations of the graph. As previously noted, this is a desirable property for activation functions in GNNs. Another important property of median filters is that they are differentiable with respect to their parameters. This is instrumental in facilitating their training, as we explain in Section \ref{sec:bp}. 

\begin{remark}\label{rmk_choice_of_shift} To keep the presentation simple, we have used the same shift operator in the definition of linear convolutional filters and median convolutional filters. This is not necessary. In our numerical experiments, we use weighted adjacency matrices, namely $\bbS=\bbA$, for the linear convolution, but we add an identity to the unweighted adjacency, namely $\bbS=\bbI + \bbN = \bbM$, for the median convolution. This choice of shift operator for the median convolution makes the neighborhood sets in \eqref{eqn_def_median_filter} nested. This is not necessarily true if we just choose the adjacency matrix as a shift operator, and it also makes neighborhoods more interpretable. Regardless of interpretation, we have observed that this mix of shift operators reduces test error. Further note that different shift operators associated with the same graph can be used at different layers or even for different features at the same layer. We have not seen advantages associated with this expansion of the representation space. \end{remark}

\begin{remark}[Alternative median graph filter definitions] The definition of median graph filters in \eqref{eqn_def_median_filter} does not make use of the shift operator weights because these are not used in the definition of the median operator in \eqref{eqn_def_median_operator}. These weights would be easy to incorporate. E.g., we could replace $x_j$ by $s_{ij}x_j$ in the computation of the median in \eqref{eqn_def_median_operator}. Or, more attuned to classical definitions of median filters, the signal value $x_j$ can be repeated within the neighborhood set a number of times proportional to the weight $w_{ij}$. Further note that \eqref{eqn_def_median_filter} performs a linear combination of different neighborhood medians. We could think of replacing this linear combination by a weighted median operation as well. All of these generalizations are possible and would retain the invariance claimed in Proposition \ref{nl_prop_invariance}. Their analysis and evaluation are beyond the scope of this paper. We refer interested readers to \cite{segarra16-globalsip, segarra17-camsap} for a comprehensive analysis of median filters for signals supported on graphs. \end{remark}

%
\subsection{Max GNNs} \label{sbs:neigh_max}

Consider once again the set $\ccalX=\{x_1,\ldots, x_n\}$, and order the elements of $\ccalX$ so that $x_{[1]} \leq x_{[2]} \leq \ldots \leq x_{[n]}$. The max of the set $\ccalX$ is the last element of the ordered sequence, which we write
\begin{equation} \label{eqn_set_max}
   \text{max}(\ccalX) = x_{[n]}\ .
\end{equation}

If the elements of $\ccalX$ are the values of a graph signal $\bbx \in \mbR^{N}$ on the nodes of a graph $\ccalG$, it makes sense to define a max operator that takes the topology of $\ccalG$ into account. In particular, this will be necessary to define activation functions in terms of max graph filters. We thus define the max operator associated with the graph shift operator $\bbS$ as follows.

%
\begin{definition}[Max operator] \label{def_max_operator}
Let $\bbS$ be a graph shift operator and denote by $\ccalN_i$ the neighborhood induced by $\bbS$ at node $i$, with $1 \leq i \leq N$. The output of the max operator $\text{max} (\bbS, \cdot)$ applied to the graph signal $\bbx$ is the graph signal $\bbz  := \text{max} (\bbS, \bbx)$, whose components we write
\begin{equation} \label{eqn_def_max_operator}
   \big[\bbz\big]_i
         = \big[ \text{max} (\bbS, \bbx)\big]_i
         =       \text{max} \big (\{ x_j : j\in\ccalN_i \}\big)\ .
\end{equation} \end{definition}

%
The max operator in Definition \ref{def_median_operator} replaces the value of $x_i$ at each node by the maximum of the values of $\bbx$ in the corresponding neighborhood $\ccalN_i$. As such, $\text{max} (\bbS, \cdot)$ acts as a nonlinear diffusion operator, akin to the median operator introduced in Definition \ref{def_median_operator}.

With the definition of the max operator at hand, we can now define multiresolution max graph filters, which are linear combinations of max operators acting on neighborhoods of different resolutions. We write them as follows. 

%
\begin{definition}[Multiresolution max graph filter] \label{def_max_filter}
Given a graph shift operator $\bbS$ and a vector of $K+1$ filter coefficients $\bbw=[w_0, \ldots, w_K]^\Tr$, the output of the max graph filter with coefficients $\bbw$ applied to the signal $\bbx$ is the signal $\bbz$,
\begin{equation} \label{eqn_def_max_filter}
   \bbz  := \sum_{k=0}^{K} w_k\, \text{max} (\bbS^k, \bbx)\ .
\end{equation} \end{definition}

%
\blue{As was the case for the median graph filter in Definition \ref{def_median_filter}, in the max graph filter in Definition \ref{def_max_filter} the summands $\text{max} (\bbS^k, \bbx)$ combine information
in the same multiresolution fashion as the polynomial terms of LSI-GFs. To see this, notice that $\text{max} (\bbS^k, \bbx)$ takes as arguments the values of the graph signal on
the same sets of $k$-hop neighbors in which the LSI-GFs from \eqref{eqn:lsigf} calculate weighted averages of the graph signal.}
However, while in LSI-GFs the signal components are averaged, in the graph max filter they are mapped to a single value by the nonlinear max operation. We can think of \eqref{eqn_def_max_filter} as another type of nonlinear convolutional filter, this one constructed from max operations. 

The max graph filter from Definition \ref{def_max_filter} is used to define a second nonlinear activation function ---the max activation function--- by considering a collection of max filter coefficients $\bbw_{\ell}^f$ and replacing the pointwise activation functions in \eqref{eqn:pointwise_nonlinearity} with max filters,
\begin{equation} \label{eqn_max_activation}
   \bbx_{\ell}^{f} 
       = \sum_{k=0}^{K}w_{\ell k}^f\,\text{max}(\bbS^k,\bbu_{\ell}^f)\ .
\end{equation}

From \eqref{eqn_max_activation}, we can then define a GNN with max activation functions in which the $\ell$th layer is a composition of \eqref{eqn:gnn_linear_summary} and \eqref{eqn_max_activation}. We call it the max GNN.

Max GNN layers compose two types of convolutional layers, linear convolutional layers [cf. \eqref{eqn:gnn_linear_summary}] and nonlinear convolutional layers based on max filters \eqref{eqn_max_activation}. The output of a max GNN is determined by the linear filter coefficients $\ccalH$ and the max graph filter coefficients $\ccalW = \{\bbw_{\ell}^{f}\}_{\ell,f}$. These parameterizations allow us to write the max GNN as the map
\begin{equation} \label{eqn_max_GNN_transform}
   \Phi(\bbx; \bbS, \ccalH, \ccalW) 
       \ = \ \bbx_L 
       \ = \ \hby(\bbx)
\end{equation}
where $\bbx_L$ is obtained by applying $L$ max GNN layers to $\bbx$. 
\
Like their median filter counterparts, max graph filters are constructed in a way that preserves the permutation invariance of GNNs. This is formally stated and proved in Proposition \ref{max_nl_prop_invariance}.

%
\begin{proposition}\label{max_nl_prop_invariance}
Consider a graph signal $\bbx$ supported on a graph with shift operator $\bbS$. Let $\Phi(\bbx; \bbS, \ccalH, \ccalW)$ be the output of a max GNN whose linear filter coefficients are $\ccalH$ and whose nonlinearity coefficients are $\ccalW$ [cf. \eqref{eqn:gnn_linear_summary}-\eqref{eqn:GNN_transform}, \eqref{eqn_max_activation}]. If $\bbP$ is a permutation matrix, then
\begin{equation} \label{eqn_max_nl_prop_invariance}
   \Phi\Big(\bbP^\Tr\bbx; \bbP^\Tr\bbS\bbP, \ccalH, \ccalW \Big) 
      = \bbP^\Tr\Phi\Big(\bbx; \bbS, \ccalH, \ccalW \Big)
\end{equation}
i.e., the output of the max GNN is invariant to permutations.
\end{proposition}
%
%
\begin{proof}
This follows from Proposition \ref{prop_invariance} and from the invariance to permutations of the max activation function in \eqref{eqn_max_activation}. 
Let $\bbS' = \bbP^\Tr\bbS\bbP$ be a graph permutation and let $\bbx' = \bbP^\Tr \bbx$ be the corresponding permuted graph signal.
\blue{This proof mimics the proof of Proposition \ref{nl_prop_invariance}, where we now use the fact that the maximum of a graph signal within its node neighborhoods is permutation invariant to show
\begin{equation*} 
   \bbz' \  = \ \sum_{k=0}^{K} \omega_k \text{max} (\bbS'^k, \bbx')
         \    \  = \ \bbP^\Tr \bbz\ .
\end{equation*}}

%
%
%
%
As was the case for median GNNs, the result above will also hold for max GNNs with multiple layers.
\end{proof}

Like their median counterparts, and unlike pointwise activation functions, max activations give GNNs the ability to extract nonlinear local features that cannot be extracted by application of linear graph convolutions alone. As seen in Proposition \ref{nl_prop_invariance}, the features extracted by max graph filters are also invariant to permutations of the graph. Finally, max filters are differentiable with respect to their parameters, which allows max activation functions to be trained (cf. Section \ref{sec:bp}).

\begin{remark}[Pooling] \label{rmk:pooling}
In deep neural network architectures, the representation dimension increases with the depth of the network and the number of features. To keep the representation dimension under control, many architectures implement pooling as an intermediate step between the convolutional filter banks and the nonlinearity.
Pooling is a summarizing two-step operation that reduces dimensionality by first computing local summaries of the signal within the graph equivalent of a window and then subsampling it. The summarizing operation can be either linear or nonlinear; the most common examples are average pooling and max-pooling. Because the graph windows on which it operates can be rescaled, pooling is also used for feature extraction at multiple resolutions. Permutation invariance is preserved if the subsampling operation is based on topological features of the graph such as the node degrees \cite{gama18-gnnarchit}. We also note that the nonlinear activation function and the pooling summarizing operation can be composed into one single localized activation function that precedes subsampling. Therefore, all design strategies of localized activation functions described herein might also be useful in designing summarizing functions in pooling operations. In the case of GNNs, pooling strategies have been proposed in \blue{\cite{gama18-gnnarchit,defferrard17-cnngraphs, bruna14-deepspectralnetworks, henaff2015deep}}.
\end{remark}

%% file: training-nonlinear.tex


In contrast with pointwise nonlinearities, localized activation functions act on multiple nodal components at a time. As a result, the computations that the NN carries out both during the forward and backward passes of data need to be updated. In this section, we look at how localized activation functions affect the gradient updates in backpropagation training (Section \ref{sbs:backpropagation}) and discuss the additional computational complexity incurred by these operators (Section \ref{sbs:complexity}).

%

\begin{figure*}[t]
	\centering
	\begin{subfigure}{.32\textwidth}
		\centering
		\includegraphics[width=\textwidth]{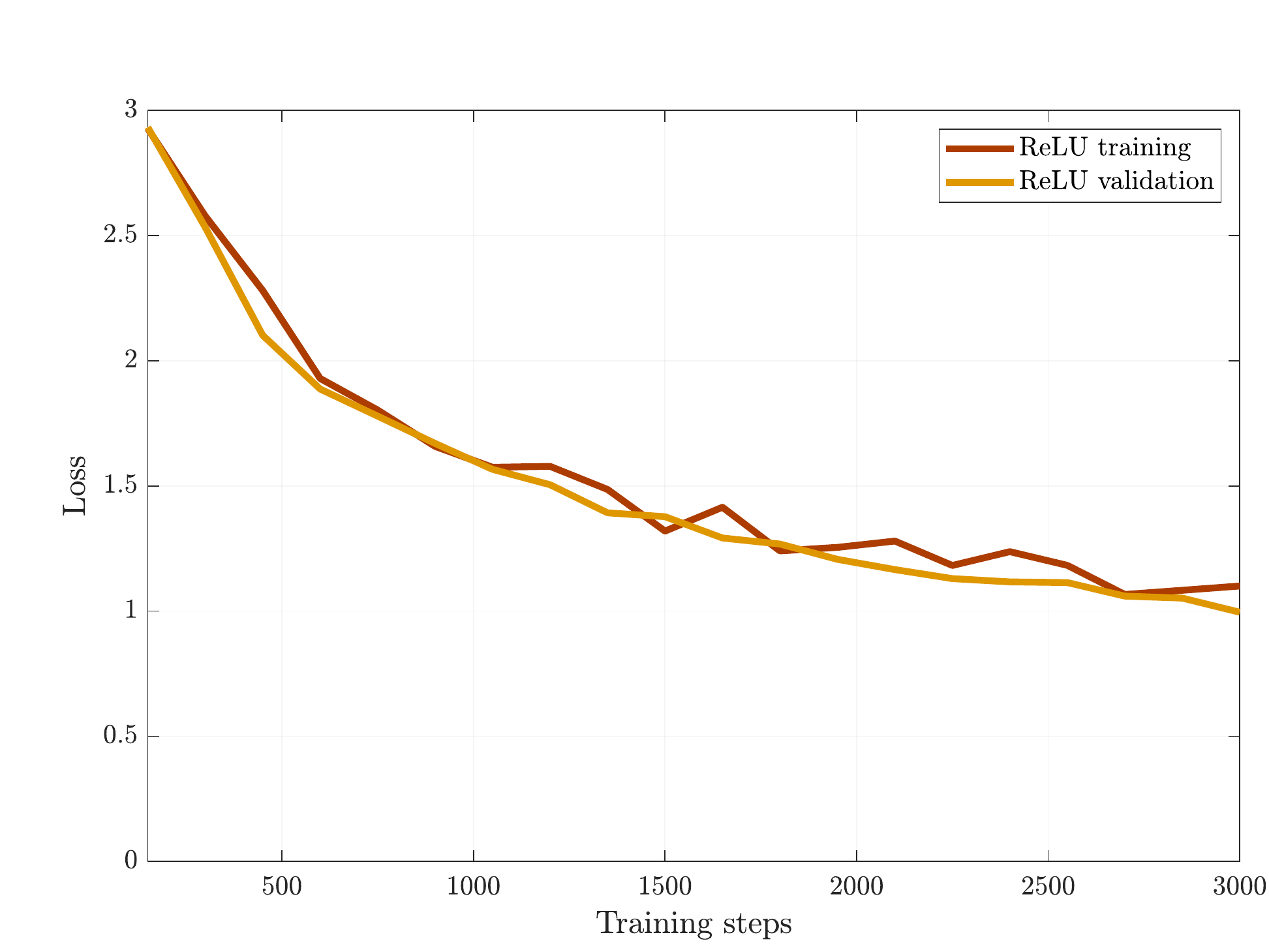}
		\caption{ReLU}
		\label{relu_val_training_loss}
	\end{subfigure}
	\hfill
	\begin{subfigure}{.32\textwidth}
		\centering
		\includegraphics[width=\textwidth]{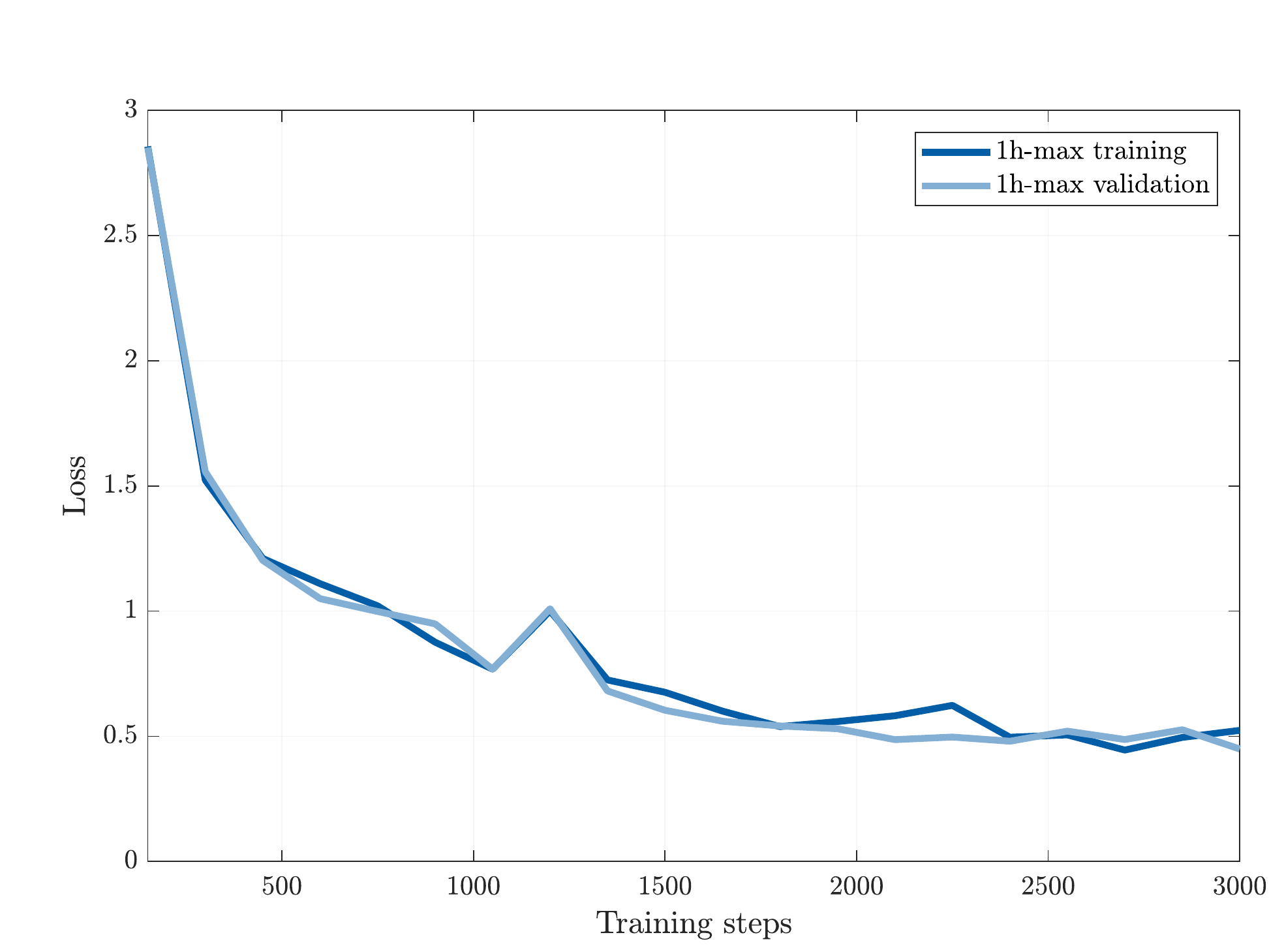} 
		\caption{1h-max}
		\label{1hmax_val_training_loss}
	\end{subfigure}
	\hfill
	\begin{subfigure}{.32\textwidth}
		\centering
		\includegraphics[width=\textwidth]{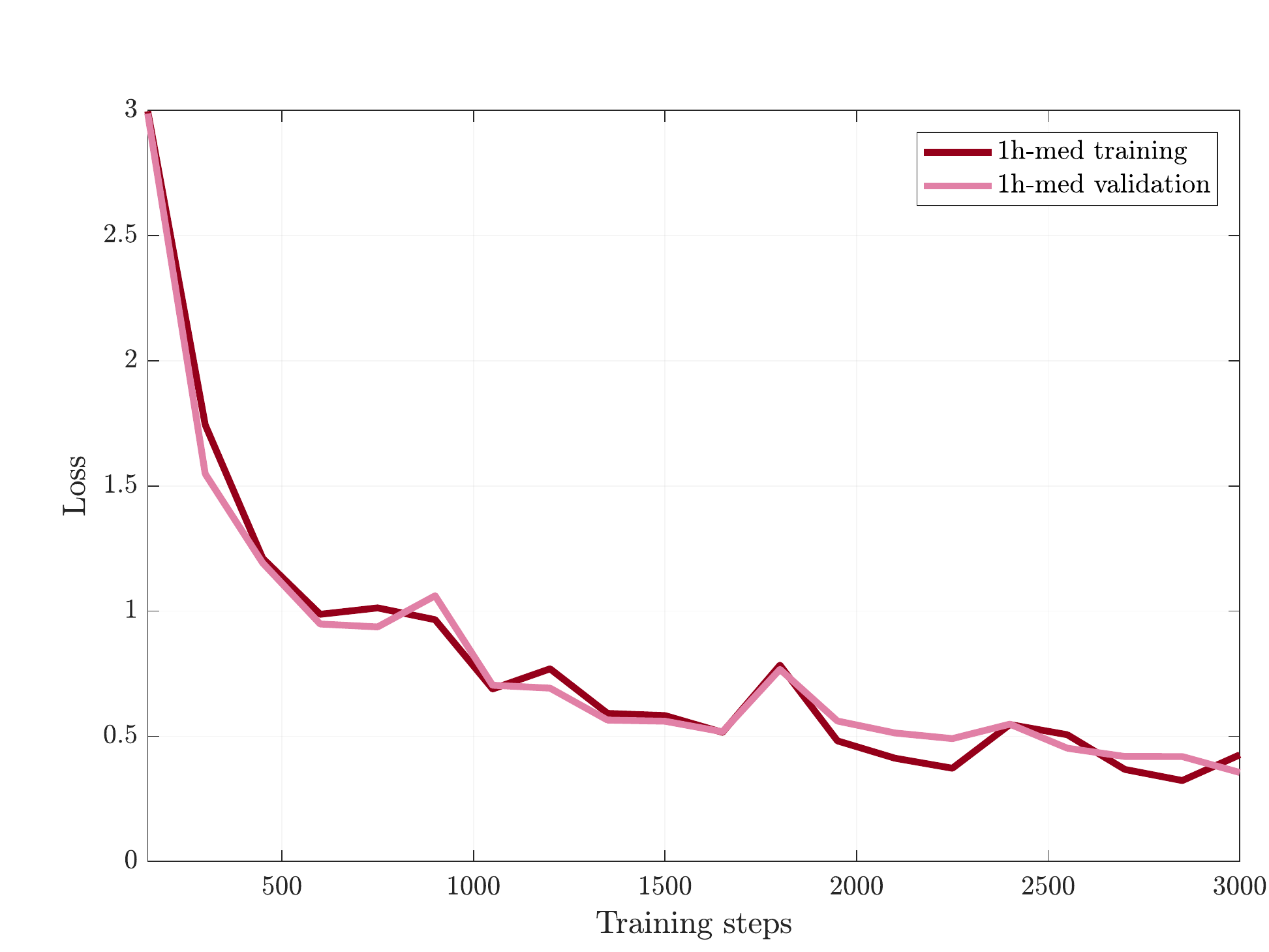} 
		\caption{1h-median}
		\label{1hmed_val_training_loss}
	\end{subfigure}
	\caption{Training and validation losses for ReLu \subref{relu_val_training_loss}, 1-hop max \subref{1hmax_val_training_loss} and 1-hop median activations \subref{1hmed_val_training_loss} \blue{in a source localization experiment with $20$ possible source nodes (classes) on an  Erd\H{o}s-R\'{e}nyi graph with $100$ nodes and edge probability $p=0.4$. The losses depicted correspond to the cross-entropy loss over $30$ training epochs. }}
	\label{fig:lossandacc}
\end{figure*}	


\subsection{Backpropagation} \label{sbs:backpropagation}

For ease of exposition, let us consider a single-layer GNN with a graph convolutional layer given by [cf. \eqref{eqn:gnn_linear_summary}]
\begin{equation} \label{eqn:singleConvLayer}
    \bbu^{f} (\bbx)
        = \sum_{g=1}^{G} \left( \bbh^{fg} \ast_{\bbS} \bbx^{g} \right)
        = \sum_{g=1}^{G} \sum_{k=0}^{K-1} h_{k}^{fg} \bbS^{k} \bbx^{g}
\end{equation}
where the input has $G$ features $\bbx^{g}$, $g=1,\ldots,G$ and the output has $F$ features $\bbu^{f}$, $f=1,\ldots,F$. The set $\ccalH = \{h_{k}^{fg}\}_{f,g,k}$ groups the corresponding trainable parameters. In median and max GNNs, the linear operation is followed by a local activation function [cf. \eqref{eqn_median_activation}, \eqref{eqn_max_activation}]
\begin{equation} \label{eqn:singleLocalActivation}
    \bbz^{f}(\bbx) = \sum_{k'=0}^{K'} w_{k'}^{f} \bbsigma\left(\bbS^{k'}, \bbu^{f}(\bbx) \right)
\end{equation}
where $\bbsigma(\bbS^{k'},\cdot): \reals^{N} \to \reals^{N}$ represents the chosen local activation function ---either the median \eqref{eqn_median_activation} or max filters \eqref{eqn_max_activation}--- with trainable parameters $\ccalW = \{w_{k'}^{f}\}_{f,k'}$. The estimate $\hby(\bbx) = [\hby^{1}(\bbx)^{\Tr},\ldots,\hby^{F}(\bbx)^{\Tr}]^{\Tr}$ is then the output of this layer, with $\hby^{f}(\bbx) = \bbz^{f}(\bbx)$ for each $f=1,\ldots,F$. Given a training set $\ccalT = \{(\bbx_{m}, \bby_{m})\}$, we choose parameters $\ccalH$ and $\ccalW$ that minimize some total loss function $\bbarJ$ over $\ccalT$
\begin{equation} \label{eqn:lossFunctionTraining}
    \bbarJ\big(\bby,\hby(\bbx) \big) = \sum_{\ccalT} J\big(\bby_{m}, \hby(\bbx_{m})\big)\ .
\end{equation}

In conventional GNNs, this is done by using the backpropagation algorithm. In the following proposition, we show that local activation functions are equally amenable to training via backpropagation, and give closed form expressions for the trainable parameters' gradient updates.

%
\begin{proposition}\label{prop:backprop_local}
Let $\ccalT = \{(\bbx_{m}, \bby_{m})\}$ be a training set. Let $\bbarJ(\bby, \hby(\bbx))$ be the loss function in \eqref{eqn:lossFunctionTraining}, where $\hby(\bbx) = [\hby^{1}(\bbx)^{\Tr},\ldots,\hby^{F}(\bbx)^{\Tr}]^{\Tr}$ is the output estimate of the single-layer GNN with convolutional layer \eqref{eqn:singleConvLayer} followed by the local activation layer \eqref{eqn:singleLocalActivation}, with $\hby^{f}(\bbx) = \bbz^{f}(\bbx)$ for each $f=1,\ldots,F$. At each round of the backpropagation algorithm, the learnable parameters $\ccalH = \{h_{k}^{fg}\}_{f,g,k}$ and $\ccalW = \{w_{k'}^{f}\}_{f,k'}$ are updated by calculating the derivatives
\begin{align}
    \frac{\partial \ \bbarJ}{\partial h_{k}^{fg}} 
        & = \sum_{\ccalT} \frac{\partial \ J}{\partial \hby^{f}} 
            \sum_{k'=0}^{K'}  
                w_{k'}^{f} \bbP_{k'}
                \ \bbS^{k} \bbx_{m}^{g} 
          \label{eqn:partialH} \\
    \frac{\partial \ \bbarJ}{\partial w_{k'}^{f}}
        & = \sum_{\ccalT} \frac{\partial \ J}{\partial \hby^{f}} 
            \bbsigma \left(\bbS^{k'}, \bbu^{f}(\bbx_{m}) \right)
          \label{eqn:partialW}
\end{align}
where $\partial J/\partial \hby^{f} = [\partial J/\partial [\hby^{f}]_{1},\ldots, \partial J /\partial[\hby^{f}]_{N}] \in \reals^{1 \times N}$ is the gradient of the loss function $J$, and
with $\bbP_{k'} \in \reals^{N \times N}$ a binary matrix such that $[\bbP_{k'}]_{ij} = 1$ if node $j$ realizes the local activation $\bbsigma(\bbS^{k'}, \cdot)$ for node $i$, and zeros for every other $j$.
\end{proposition}
%
%
\begin{proof}
Let us start by proving \eqref{eqn:partialH}. We take the derivative of the total loss function $\bbarJ$ with respect to the scalar coefficient $h_{k}^{fg}$. Towards this end, denote the gradient of $J$ with respect to all the features by $\partial J/\partial \hby = [\partial J/\partial \hby^{1},\ldots,\partial J/\partial \hby^{F}] \in \reals^{1 \times NF}$ where each $\partial J / \partial \hby^{f} \in \reals^{1 \times N}$ is the gradient with respect to feature $f$. Denote the gradient of the output of the GNN with respect to the specific parameter $h_{k}^{fg}$ by $\partial \hby/\partial h_{k}^{fg} = [(\partial \hby^{1}/\partial h_{k}^{fg})^{\Tr},\ldots,(\partial \hby^{F}/\partial h_{k}^{fg})^{\Tr}]^{\Tr} \in \reals^{NF \times 1}$, where $\partial \hby^{f'}/\partial h_{k}^{fg} \in \reals^{N \times 1}$ for each $f'=1,\ldots,F$. We note that $\partial \hby^{f'}/\partial h_{k}^{fg} = \bbzero$ whenever $f' \neq f$. Applying the chain rule once, we get
\begin{equation} \label{eqn:dLdh}
    \frac{\partial \ \bbarJ}{\partial h_{k}^{fg}} 
        = \sum_{\ccalT}
            \frac{\partial \ J}{\partial \hby} 
            \left. \frac{\partial \ \hby }{\partial h_{k}^{fg}}\ \right|_{\bbx_{m}}
        = \sum_{\ccalT}
            \frac{\partial \ J}{\partial \hby^{f}} 
            \left. \frac{\partial \ \hby^{f} }{\partial h_{k}^{fg}}\ \right|_{\bbx_{m}}
\end{equation}
since $\partial \hby^{f'}/\partial h_{k}^{fg} = \bbzero$ whenever $f' \neq f$. Now, we focus on $\partial \hby^{f}/\partial h_{k}^{fg}$, and using \eqref{eqn:singleLocalActivation} and the chain rule once again, we get
\begin{equation} \label{eqn:dydh}
\frac{\partial \ \hby^{f}}{\partial h_{k}^{fg}}
= \sum_{k'=0}^{K'} w_{k'}^{f} 
\frac{\partial \ \bbsigma(\bbS^{k'},\cdot)}{\partial \bbu^{f}} 
\frac{\partial \ \bbu^{f}}{\partial h_{k}^{fg}}
\end{equation}
where $\partial \bbsigma(\bbS^{k'},\cdot)/\partial \bbu^{f} \in \reals^{N \times N}$ is the Jacobian, with $[\partial \bbsigma/\partial \bbu^{f}]_{ij} = \partial [\bbsigma]_{i}/\partial [\bbu^{f}]_{j}$ being the corresponding (sub-)derivative. The local activation function $\bbsigma(\bbS^{k'},\bbu^{f})$ outputs a graph signal where each element is a nonlinear combination of the $k'$-hop neighbors of each node. In both cases (median and max), the function application is actually equivalent to selecting a value from the ones at the neighboring nodes (i.e. the output of $[\bbsigma(\bbS^{k'},\bbu^{f})]_{i}$ is the value of $[\bbu^{f}]_{j}$ for some $j \in \ccalN_{i}^{k'}$). Therefore, the (sub-)derivative of $[\bbsigma]_{i}$ with respect to the input graph signal $\bbu^{f}$ is a vector with a $1$ in the position corresponding to the nodes $j \in \ccalN_{i}^{k'}$ that generate the output of the function, and zeros elsewhere. This implies that $\partial \bbsigma(\bbS^{k'},\cdot)/\partial \bbu^{f} = \bbP_{k'} \in \reals^{N \times N}$ is a binary matrix with $[\bbP_{k'}]_{ij} = 1$ if node $j$ generates the output of the function at node $i$ and $0$ for all other $j$. Finally, since per \eqref{eqn:singleConvLayer} $\bbu^{f}$ is a linear function of $h_{k}^{fg}$,
\begin{equation} \label{eqn:dudh}
\frac{\partial \ \bbu^{f}}{\partial h_{k}^{fg}} = \bbS^{k} \bbx^{g}\ .
\end{equation}
Using \eqref{eqn:dudh} together with the fact that $\partial \bbsigma(\bbS^{k'},\cdot)/\partial \bbu^{f} = \bbP_{k'}$ back in \eqref{eqn:dydh}, and, consequently, back in \eqref{eqn:dLdh}, completes the proof of \eqref{eqn:partialH}.

%

\begin{figure*}[t]
	\centering
	\begin{subfigure}{.4\textwidth}
		\centering
		\includegraphics[width=\textwidth]{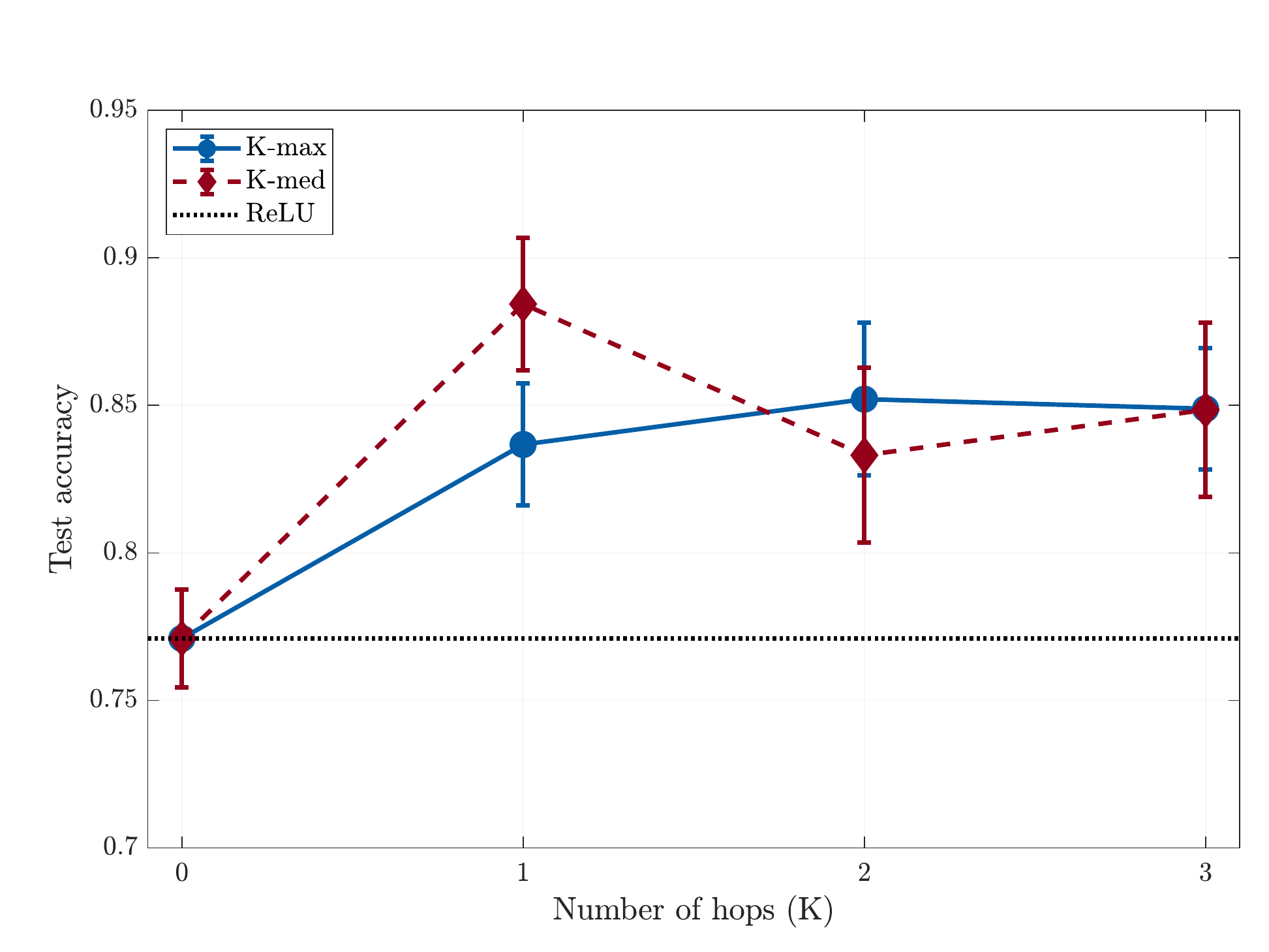}
		\caption{}
		\label{ER}
	\end{subfigure}
	\hspace{5em}
	\begin{subfigure}{.4\textwidth}
		\centering
		\includegraphics[width=\textwidth]{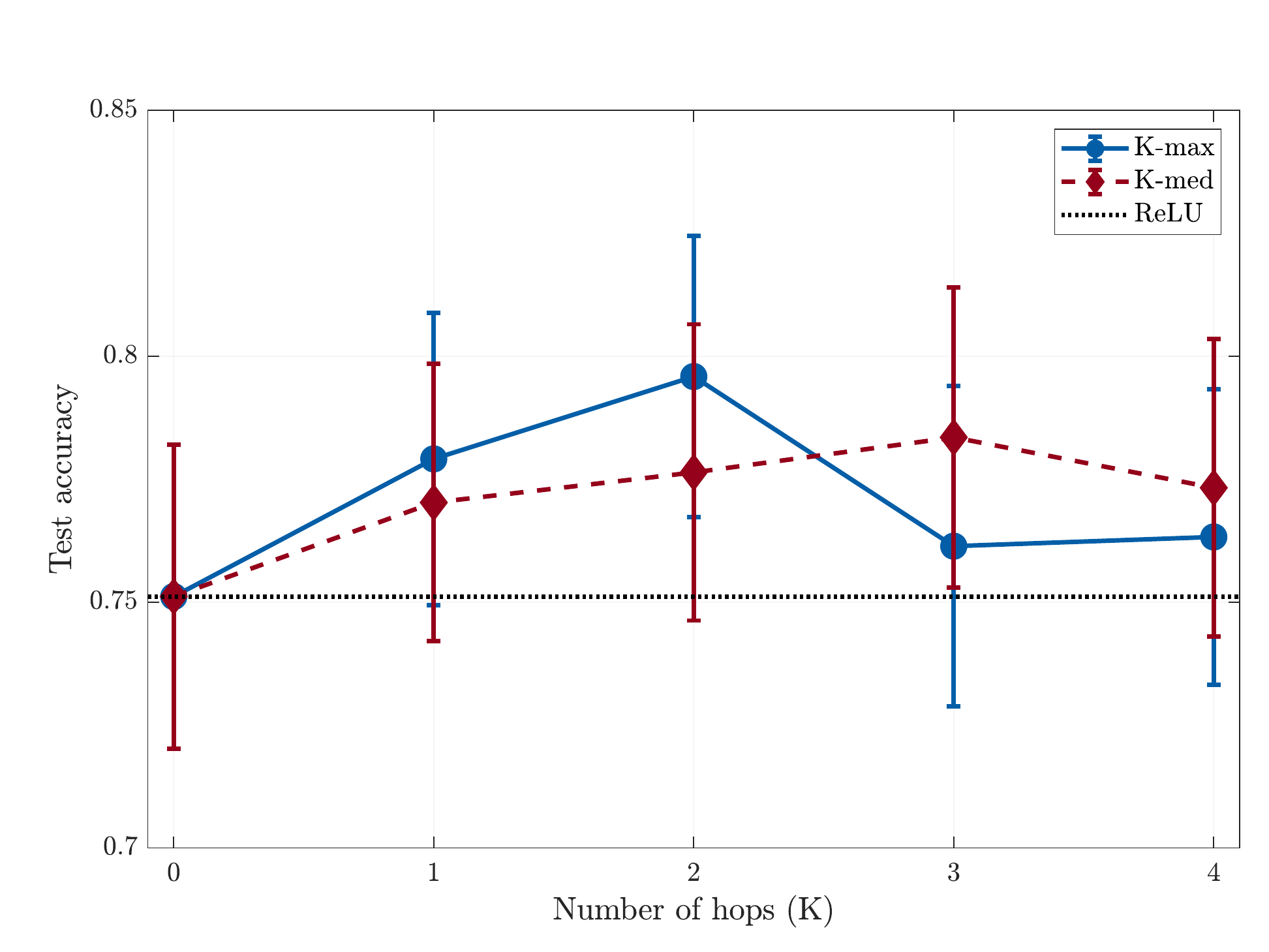} 
		\caption{}
		\label{geom}	
	\end{subfigure}
	\caption{Source localization test accuracy in GNN architectures with localized activation functions, by hop. \subref{ER} Average accuracy and standard error on 40 ER graphs with edge probability $p = 0.4$. \subref{geom} Average accuracy and standard error on 40 geometric graphs with radius $d = 0.15$ on the unit square. The error bars were scaled by 0.2.}
		\label{fig:byhop}
	\end{figure*}	

For \eqref{eqn:partialW} we have
\begin{equation} \label{eqn:dLdw}
    \frac{\partial \ \bbarJ}{\partial w_{k'}^{f}} 
        = \sum_{\ccalT}
            \frac{\partial \ J}{\partial \hby} 
            \left. \frac{\partial \ \hby}{\partial w_{k'}^{f}}\right|_{\bbx_{m}} 
        = \sum_{\ccalT}
            \frac{\partial \ J}{\partial \hby^{f}} 
            \left. \frac{\partial \ \hby^{f}}{\partial w_{k'}^{f}}\right|_{\bbx_{m}}
\end{equation}
where we used again the fact that $\partial \hby^{f'}/\partial h_{k}^{fg}=\bbzero$ whenever $f' \neq f$, and where the derivative of the feature $f$ with respect to the nonlinear weight parameter is denoted by $\partial \hby^{f}/\partial w_{k'}^{f} = [(\partial [\hby^{f}]_{1}/\partial w_{k'}^{f})^{\Tr},\ldots,(\partial [\hby^{f}]_{N}/\partial w_{k'}^{f})^{\Tr}] \in \reals^{N \times 1}$. Observe that $\hby^{f}$ is a linear function of $w_{k'}^{f}$, and so using \eqref{eqn:singleLocalActivation} we get
\begin{equation} \label{eqn:dydw}
   \left. \frac{\partial \ \hby^{f}}{\partial w_{k'}^{f}} \right|_{\bbx_{m}} = \bbsigma \left(\bbS^{k'}, \bbu^{f} (\bbx_{m}) \right)\ .
\end{equation}
Replacing \eqref{eqn:dydw} back in \eqref{eqn:dLdw} proves \eqref{eqn:partialW}.
\end{proof}

It should be noted that both in \eqref{eqn:partialH} and \eqref{eqn:partialW} the gradient updates depend on quantities that are either available from the start ---like $\bbS$--- or made available in the forward pass immediately before the backpropagation step. In particular, the weights $w_{k'}^f$ in \eqref{eqn:partialH} are initialized before the first forward and backward passes.

%
\subsection{Computational complexity} \label{sbs:complexity}

In this section, we discuss and quantify the additional computational complexity incurred by localized activation functions in the forward and backward passes of data needed to train the single layer median/max GNNs of the previous subsection. These GNNs are compared with a single-layer GNN containing only pointwise activation functions.

%
\begin{corollary}\label{cor:backprop_pointwise}
In Proposition~\ref{prop:backprop_local}, replace the local activation function layer \eqref{eqn:singleLocalActivation} by a pointwise activation function $\sigma: \reals \to \reals$,
\begin{equation} \label{eqn:singlePointwiseActivation}
    \bbz^{f} = \bbsigma(\bbu^{f})
\end{equation}
where $[\bbsigma(\bbu^{f})]_{i} = \sigma([\bbu^{f}]_{i})$ for all $i=1,\ldots,N$. In this case, the derivatives used in the backpropagation algorithm are
\begin{equation} \label{eqn:partialHpointwise}
    \frac{\partial \ \bbarJ}{\partial h_{k}^{fg}} 
        = \sum_{\ccalT} \frac{\partial \ J}{\partial \hby^{f}} 
            \diag(\bbq)
            \ \bbS^{k} \bbx_{m}^{g}
\end{equation}
where $\bbq \in \reals^{N}$ with $[\bbq]_{i} = d\sigma/dx \ |_{x = [\bbu^{f}]_{i}}$.
\end{corollary}
%
%
\begin{proof}
We set $K'=0$ in \eqref{eqn:partialH} in Proposition \ref{prop:backprop_local}, since we only consider the value of the signal at the node individually and $w_{k'}^{f}=1$ (the output of conventional pointwise activation functions is not modified by any trainable weights). In this scenario, the matrix $\bbP_{k'} = \bbP_{0}$ takes the form of a diagonal matrix, because the only nodes contributing to the output of the nonlinearity are the nodes themselves. The derivatives of the nonlinearity $\sigma$ at each node are then the diagonal values of the matrix $\bbP_{0} = \diag(\bbq)$ with $[\bbq]_{i} = d\sigma/dx \ |_{x = [\bbu^{f}]_{i}}$.
\end{proof}

The overall complexity of the activation step in \emph{one} forward pass of the GNN with pointwise activations is $\ccalO(N)$, that is, one nonlinearity per node. In a forward pass of the GNN with localized activation functions, the added complexity stems mainly from the sorting operations performed across multiple neighborhoods of a given node to find the maximum or the median of the graph signal in those regions. As a result, the multi-hop maximum and the multi-hop median activations incur an overall worst-case computational complexity of $\ccalO (NK'd^{K'}\log (d^{K'}))$, with $K'>0$ and $d$ denoting the maximum degree. Although the activation function reach $K'$ is usually small, it exponentiates $d$, which need not be. The graph degree is therefore bound to be a considerable limiting factor complexity-wise, but this is expected since our nonlinearities are now local instead of pointwise. In any case, we note that\blue{, as long as the degree is not a function of the number of nodes (this is the case, for example, of regular networks and of most small-world networks \cite[Ch. 10]{newman2018networks}),} the complexity is still linear in $N$.

In the backward pass, the GNN with pointwise activations incurs $\ccalO(N)$ operations (cf.\eqref{eqn:partialHpointwise}) to compute the derivative with respect to $h_{k}^{fg}$. In the case of the GNN with local activation functions, following Proposition \ref{prop:backprop_local} we have $\ccalO(K'N)$ for the computation of the derivative with respect to $h_{k}^{fg}$ followed by $\ccalO(N)$ operations to compute the derivatives with respect to $w_{k'}^{f}$. These results are summarized in Table \ref{table_comp}.

\begin{table}
	\centering
\begin{tabular}{lcc} \hline
		& \multicolumn{2}{c}{Complexity} \\
Activation    & Forward pass & Backpropagation	\\ \hline
Pointwise		& $\ccalO(N)$ & $\ccalO(N)$ \\
Local	& $\ccalO(NK'd^{K'}\log(d^{K'}))$ & $\ccalO(K'N)$ \\ \hline
\end{tabular}
	\caption{Complexity comparison in a single forward pass and a single backward pass of data for architectures with ReLu and localized activation functions.}
	\label{table_comp}
\end{table}

%% file: sims-nonlinear.tex


%

\begin{figure*}[t]
	\centering
	\begin{subfigure}{.4\textwidth}
		\centering
		\includegraphics[width=\textwidth]{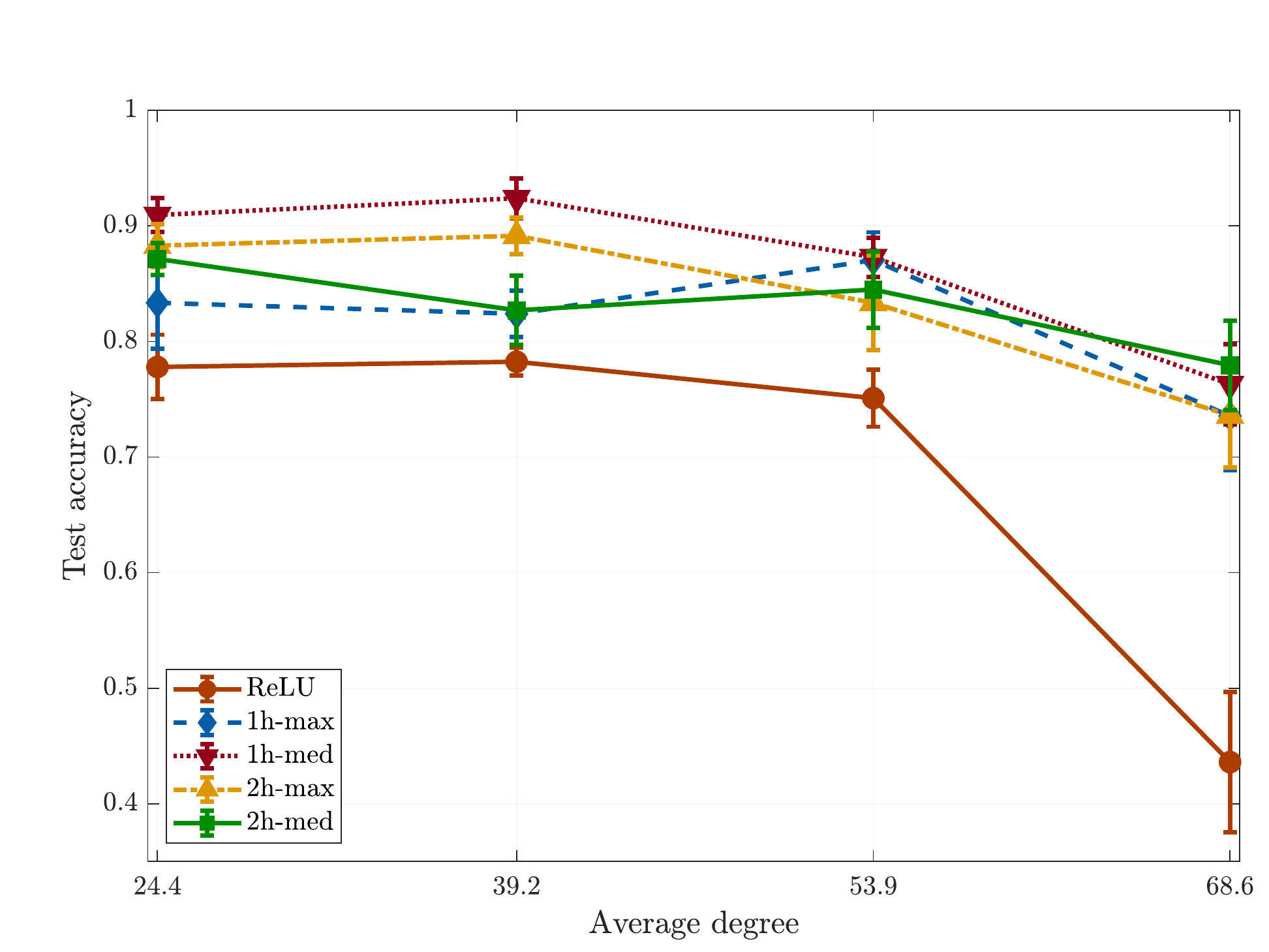}
		\caption{}
		\label{ERdeg}
	\end{subfigure}
	\hspace{5em}
	\begin{subfigure}{.4\textwidth}
		\centering
		\includegraphics[width=\textwidth]{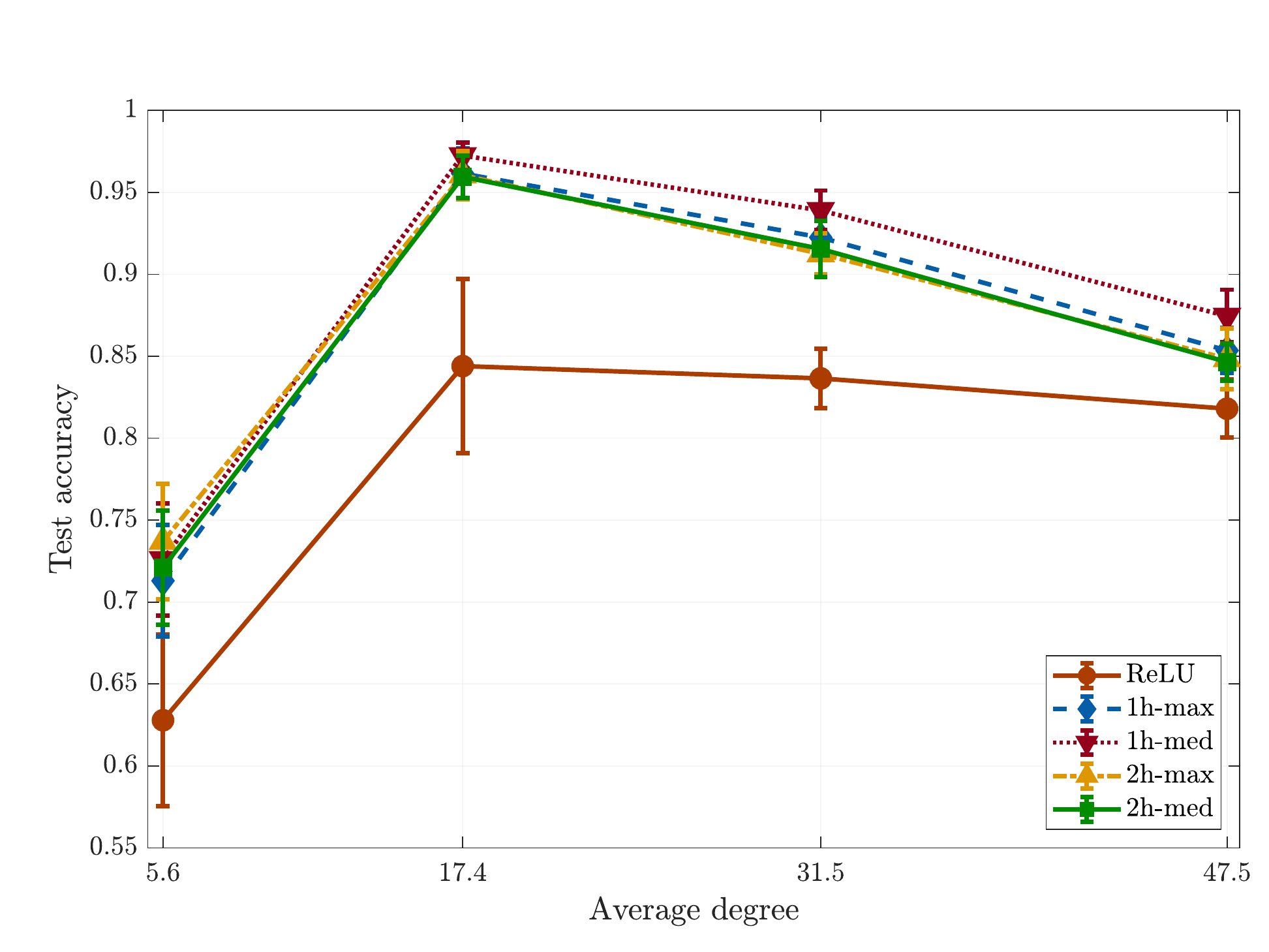} 
		\caption{}
		\label{geomdeg}		
	\end{subfigure}
	\caption{Source localization test accuracy in GNN architectures with localized activation functions, by degree. \subref{ERdeg} Average accuracy and standard error on 10 ER graphs with edge probabilities $p = 0.25, 0.40, 0.55, 0.70$. \subref{geomdeg} Average accuracy and standard error on 10 geometric graphs with radii $d = 0.15, 0.27, 0.38, 0.5$ on the unit square. The error bars were scaled by 0.2.}
	\label{fig:bydegree}
\end{figure*}	

We assess the performance of median and max GNNs in \blue{four} different applications. In all scenarios, the multi-hop maximum and the multi-hop median are compared with the ReLU. The first problem is source localization on Erd\H{o}s-R\'{e}nyi (ER) and geometric graphs. In this synthetic setting, we analyze how the activation function reach $K$ and the underlying graph degree affect classification accuracy. The second application is authorship attribution of text excerpts taken from 19th century novels, modeled as a binary classification problem detailed in Section \ref{subsec_auth}. 
\blue{In the third experiment}, we tackle the problem of predicting movie ratings using the MovieLens 100k dataset. In addition to comparing localized and pointwise activation functions, we contrast our performance with that of the recommendation systems proposed in \cite{weiyu18-movie} and \cite{monti17-movie}. \blue{The last experiment is a node classification task in which we use the Cora citation network and associated dataset to classify scientific articles into 7 different classes.}

\blue{In Sections \ref{subsec_sourceloc} through \ref{subsec_movie}, the simulated GNNs predict class labels for graph signals. In Section \ref{subsec_cora}, they predict labels for nodes of the graph. Unless otherwise noted, all models} consist of one convolutional layer with $F_1=32$ linear graph filters that have $K_1=5$ filter taps each, followed by the activation function under analysis (ReLU, $K$-hop median or $K$-hop max). The GSO of the convolutional filters is always the adjacency matrix; \blue{in Sections \ref{subsec_sourceloc}, \ref{subsec_auth} and \ref{subsec_cora}, the best results were obtained when the adjacency matrix was rescaled by the inverse of its largest eigenvalue and so these are the results that we report.} No pooling is performed, and \blue{in the graph signal classification settings (Sections \ref{subsec_sourceloc}-\ref{subsec_movie})} the convolutional layer is followed by a fully connected layer that carries out a softmax classification with $F_2 = C$ nodes. $C$ is the number of classes, which is different for each problem. In all applications, the GNNs were trained using the ADAM algorithm for stochastic optimization. This algorithm keeps an exponentially decaying average of past gradients with decaying factors $\beta_{1} = 0.9$ and $\beta_{2} = 0.999$ \cite{kingma17-adam}.	
	

\subsection{Source Localization} \label{subsec_sourceloc}

Source localization is a classification problem where we aim to identify the node that originated a diffusion process on a graph \cite{segarra16-deconv}.  Take for instance a graph $\ccalG$ with $N$ nodes and adjacency matrix $\bbW$. To be specific, let $c \in \{1,...,N\}$ represent the index of the source node and consider the seeding graph signal $\bbx_0$, which is defined as $[\bbx_0]_i = 1$ if $i = c$ and $[\bbx_0]_{i}=0$ for all other $i$. The corresponding diffused signal at time $t=1,2,\ldots$ is then $\bbx(t) = \bbW^t \bbx_0$, which satisfies $\bbx(0) = \bbx_0$. Given $\bbx(t)$ and without knowing $t$, we want to identify the node $c$.   

In all of this section's experiments, we train a GNN to predict the source node of a graph diffusion process by optimizing a cross entropy loss with 0.005 learning rate. In every round, 10,000 synthetic training samples consisting of a diffused graph signal $\bbx(t)$ at a random time $t$ (input) and its source $c$ (true label) were evaluated in batches of 100. To prevent overfitting, we set the node dropout probability to 50\% during training \cite{srivastava14-dropout}. The validation and test sets comprised 200 input-output samples in every round. 

We first analyze the evolution of the training and validation losses over 30 epochs for an ER graph with 100 nodes, edge probability 0.4 and and 20 possible sources that we choose uniformly at random. This amounts to a classification problem with 20 classes. The training vs. validation loss plots for the ReLU, the 1-hop maximum and the 1-hop median are shown in Figure \ref{fig:lossandacc}. No architecture overfits the training set, and they all achieve a comparable loss over the unseen validation loss. 

Next, we study the source localization problem on 40 random ER \cite{ErdosRenyi59-RandomGraphs} and 40 random geometric graphs \cite[Ch. 4]{penrose07-rgg}. Each graph has 100 nodes, 10 of which are randomly picked to be potential sources (10 classes). The ER graphs have edge probability 0.4 and average degree 39.4. The geometric graphs have radius 0.15 on the unit square and average graph degree 5.6. Training was done in 20 epochs for both types of graphs.

First, we analyze localized activation function performance in terms of their reach in number of hops, i.e., $K$ in expressions \eqref{eqn_median_activation} and \eqref{eqn_max_activation}. We adopt the convention that, for $K = 0$, both activation functions amount to the ReLU, $\mbox{max}\{0,x_i\}$, since expressions \eqref{eqn_median_activation} and \eqref{eqn_max_activation} are linear in $\bbx$ for $K = 0$. Figure \ref{ER} shows the average test accuracies for ER graphs. At $K=1$, localized activation functions increase accuracy of at least 6.5 percentage points over the ReLU.  However, as $K$ grows bigger, their performance appears to stagnate or decay. A plausible explanation for that is that on graphs with such high connectivity and average degree, neighborhoods become more redundant as they grow larger (larger $K$).

The average test accuracies as a function of $K$ are displayed in Figure \ref{geom} for geometric graphs. Even if the 2-hop max showcases the largest accuracy, its performance degrades as the number of hops increases. On the other hand, the median sustains consistent improvements up to $K=3$, but its performance also decays from there to $K=4$. This explained by the fact that geometric graphs with small radii are not very connected and have an almost Euclidean structure. Thus, the more distant the neighborhood, the smallest its influence is likely to be on a node. It is natural, then, to expect a function returning extreme values like the maximum to add a lot of high intensity noise from distant neighborhoods as $K$ grows and thus degrade more abruptly in performance than a smoothing operation like the median. Nonetheless, the median still performs worse as neighborhoods increase, due to the excess of information and possible redundancies. Although all localized activation functions outperform the ReLU, they do so by smaller margins and with higher variances than those observed on ER graphs. This is further evidence that localized activation functions are less powerful on graphs with some structural regularity, which is precisely the case of geometric graphs.

The next analysis was done by keeping all the same simulation parameters and varying the graph degree. Results are presented in Figure \ref{ERdeg} for ER graphs and in Figure \ref{geomdeg} for geometric graphs. On both ER and geometric graphs, the 1-hop median delivers the best results by degree. On ER graphs, localized activation functions are consistently better than the ReLU, and even if classification accuracy decreases with the graph degree regardless of the choice of activation function, the accuracy gap between localized activations and the ReLU increases systematically.

On geometric graphs (Figure \ref{geomdeg}), the best accuracy is obtained in the middle of the degree range. The worst accuracies are observed for graphs with small degree; this is once again related to them having an almost Euclidean structure that is more likely to benefit from the CNN rather than the GNN apparatus. Localized activation functions outperform the ReLU in all scenarios. The steady decay in performance that occurs at higher degrees is accompanied by a reduction in the accuracy gap between localized and pointwise activations, in contrast with what we observed for ER graphs. This can be explained by the highly connected patterns arising from the smaller variances in the individual nodes' degrees of geometric graphs.

\subsection{Authorship attribution} \label{subsec_auth}

In this section, we assess the performance of localized activation function GNNs in an authorship attribution problem based on real data. The graphs we consider are author-specific word adjacency netowrks (WANs), which are directed graphs whose nodes are function words and whose edges represent the probability of transitioning between a particular pair of words in a text written by the author. Function words are prepositions, pronouns, conjunctions and other words with syntactic importance and little semantic meaning; their use in authorship attribution was first discussed in \cite{mosteller64-auth} and is based on the fact that their usage carries stylometric information about the author while being content independent. 

%
\begin{figure}[t]
	\centering
	\includegraphics[width=.95\columnwidth]{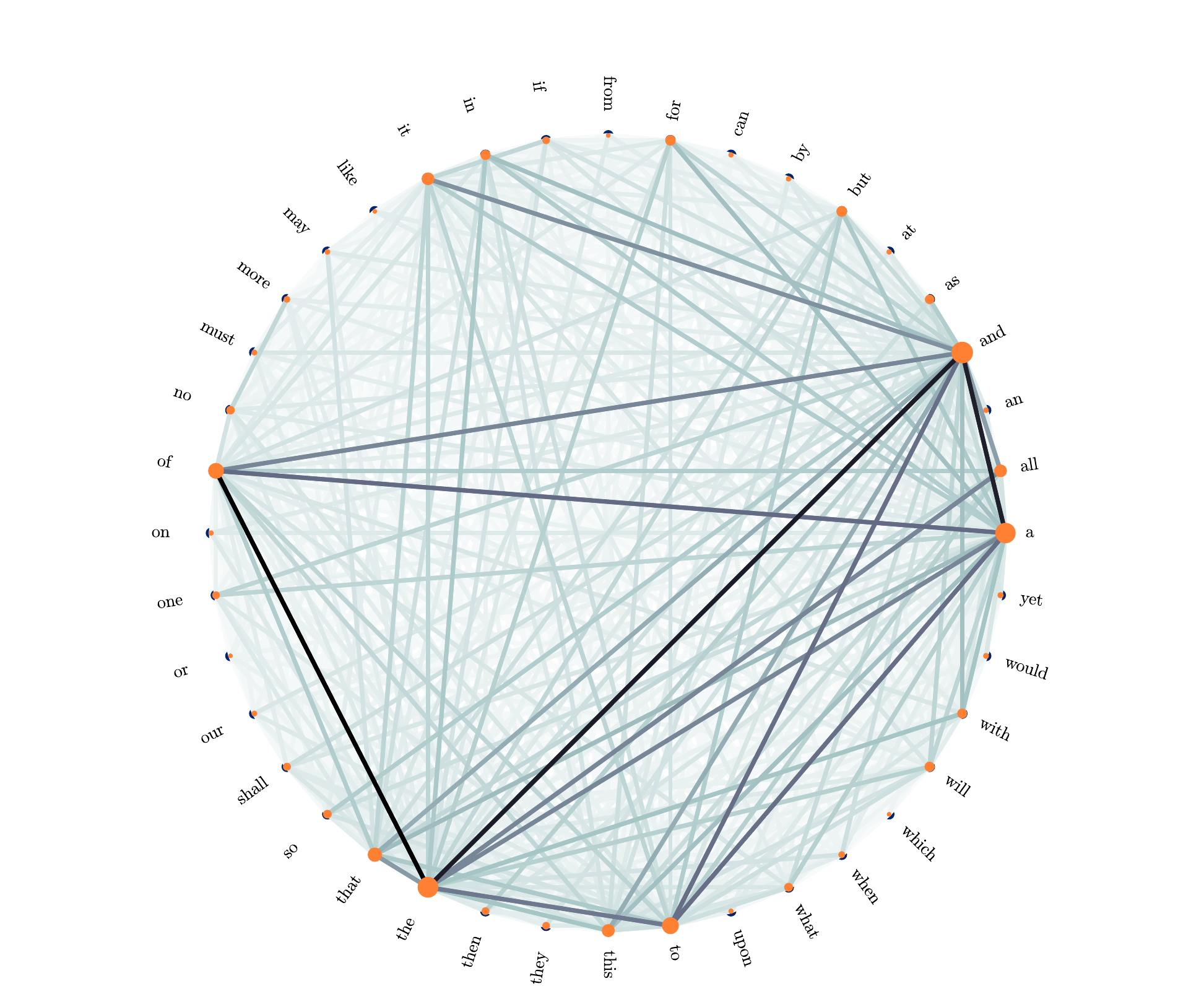}
	\caption{Example of WAN with 40 function words built from the play ``The Humorous Lieutenant'' by John Fletcher. Orange circles represent frequency with radius proportional to word count. The darker the edge color, the higher is edge weight.}
	\label{fig:authorship}
\end{figure}	

%
\begin{figure*}[t]
	\centering
	\begin{subfigure}{0.4\textwidth}
		\centering
		\includegraphics[width=\columnwidth]{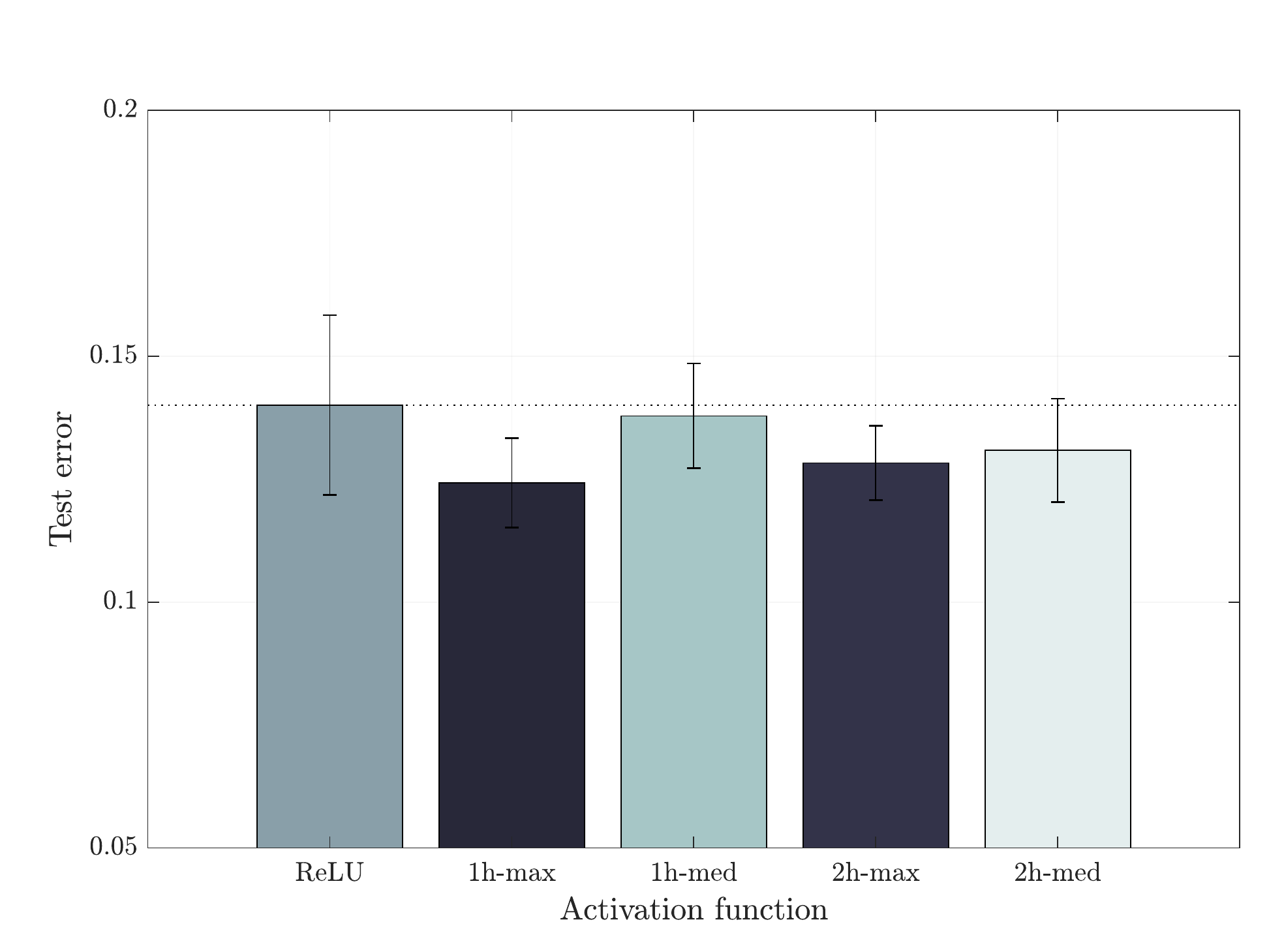}
		\caption{}
		\label{bronte}
	\end{subfigure}
	\hspace{5em}
	\begin{subfigure}{0.4\textwidth}
		\centering
		\includegraphics[width=\columnwidth]{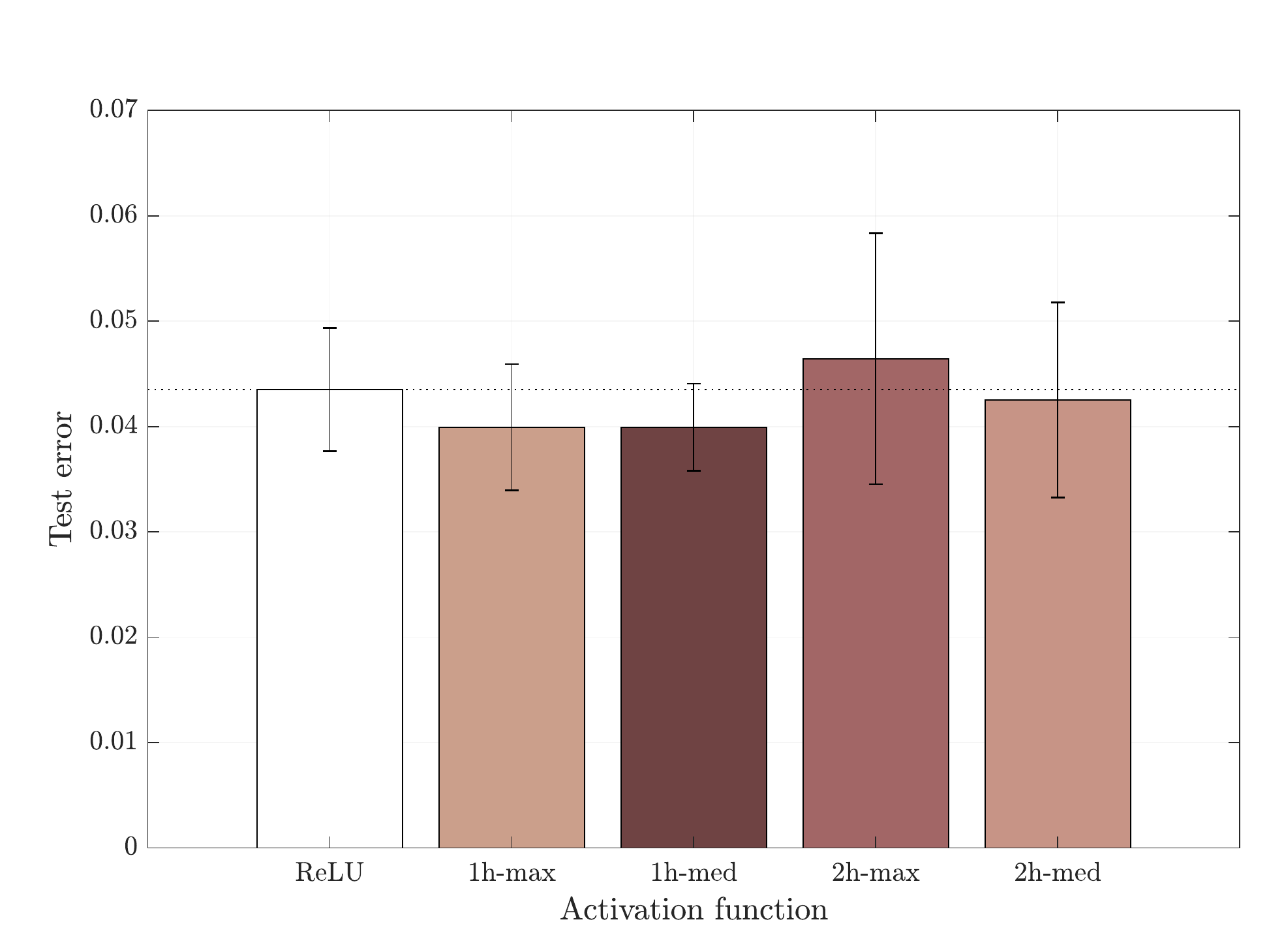}
		\caption{}
		\label{austen}
	\end{subfigure}
	\caption{Authorship attribution test error in GNN architectures with localized activation functions. Average classification error and standard deviation on 10 different training-test 80-20 splits for \subref{bronte} Br{\"o}nte and \subref{austen} Austen. The error bars were scaled by 0.5. }
	\label{fig:authorship2}
\end{figure*}	

We consider $N=211$ nodes or functions words. Using the method in \cite{eisen15-auth}, we build single-author WANs for Emily Br{\"o}nte and Jane Austen. To build each author's WAN, we process their texts to count the number of times that each pair of function words co-appear in 10-word windows. These are inputted to a $N \times N$ function-word matrix and normalized row-wise. The resulting matrix is the WAN adjacency matrix, which can also be interpreted as a Markov chain transition matrix. Because the order in which function words appear matters, the resulting graphs are directed. As for the graph signals, they are defined as each function word's count among 1,000 words. Thus, the texts available by a given author are split in 1,000-word excerpts (signals) where we store the frequency of each of the function words. 

Splitting an author's texts between training and test sets on a 80-20 ratio, each author's WAN is generated from function word co-appearance counts in the training set only. An example of such a network is depicted in Figure \ref{fig:authorship}. The graph signals in the training set are the individual function words' counts in these same excerpts and in excerpts by other authors picked at random from a pool of 21 authors to yield a balanced classification problem. Paired with a binary label where 1 indicates a text by the author in question and 0 a text by any other author in the pool, these constitute the input-output pairs used to train the GNN. Test samples are defined analogously, but we only consider excerpts that have not been used to build the author's WAN. The loss function is the cross entropy, which we optimize in 25 training epochs and batches of 20 samples, with learning rate 0.005 and without dropout.

Figure \ref{fig:authorship2} presents the authorship attribution accuracy results for Emily Br{\"o}nte (Figure \ref{bronte}) and Jane Austen (Figure \ref{austen}). Ten rounds of simulations were conducted for each author by varying the training and test splits.

The average average out-degree of the WANs built for author Emily Br{\"o}nte was 77.9. The training and test sets consisted of 1,092 and 272 1,000-word excerpts, both with equally balanced classes. For Jane Austen, the average out-degree of the WANs considered was 88.3, and the training and test sets contained 1,234 and 308 labeled excerpts respectively. 

On Figure \ref{bronte}, we see that median and max GNNs did consistently better than the ReLU GNNs on discerning between texts written by Br{\"o}nte and any other author in the pool. Although the smallest classification error, of 12.43\% (34/272), was obtained with the 1h-max, every other localized activation outperforms the ReLU on average, with significantly smaller test errors and deviations around the average.

For the author Jane Austen (Figure \ref{austen}), three of our schemes perform better than the ReLU. The 2-hop localized activations do worse than the ReLU, which could be explained by the higher average degree of this author's WANs. The gap between the best performing localized activation --- the 1-hop median --- and the ReLU is not as big as in the previous example, but it still amounts to at least one extra excerpt being labeled correctly. What is more impressive is this architecture's ability to correctly attribute text fragments as short as a single page with up to 98.37\% accuracy (303/308). This was the best observed accuracy in all 10 realizations and it was obtained by training the 1-hop max.	


\subsection{Recommender systems} \label{subsec_movie}  

\begin{table}
	\centering
	\begin{tabular}{lc}
	\hline
	Activation    & Number of parameters	\\ \hline
	ReLU		& $160$ \\
	$K$-hop median/max	& $162$ \\ \hline
	\end{tabular}
	\caption{Number of parameters in the convolutional layers of the 3 GNNs simulated in the recommender systems problem: ReLU, 1-hop median and 1-hop max.}
	\label{table_params}
\end{table} 

The third application we consider is movie rating prediction using the MovieLens 100k dataset \cite{harper16-movielens}, which contains ratings that a set of users have given to a subset of movies. There are $U=943$ users and $M=1,582$ movies (items), the ratings range from 1 to 5, only 100,000 out of 1,491,826 ratings are known and the ratings for unknown user-movie pairs are set to 0. Given an incomplete $U \times M$ rating matrix, we can define two different graphs --- an \textbf{(a)} user similarity network and a \textbf{(b)} movie similarity network. Both are constructed by computing Pearson correlations considering only \textbf{(a)} items that have been rated by pairs of users or \textbf{(b)} users that have rated the same item pairs. For a given node, we keep only the top-$k$ user or item pairs with highest similarity, which yields a directed graph. Correspondingly, a user-based and a movie-based GNN architecture can be defined atop each one of these graphs.

%
\begin{figure*}[t]
	\centering
	\begin{subfigure}{.35\textwidth}
		\centering
		\includegraphics[width=\textwidth]{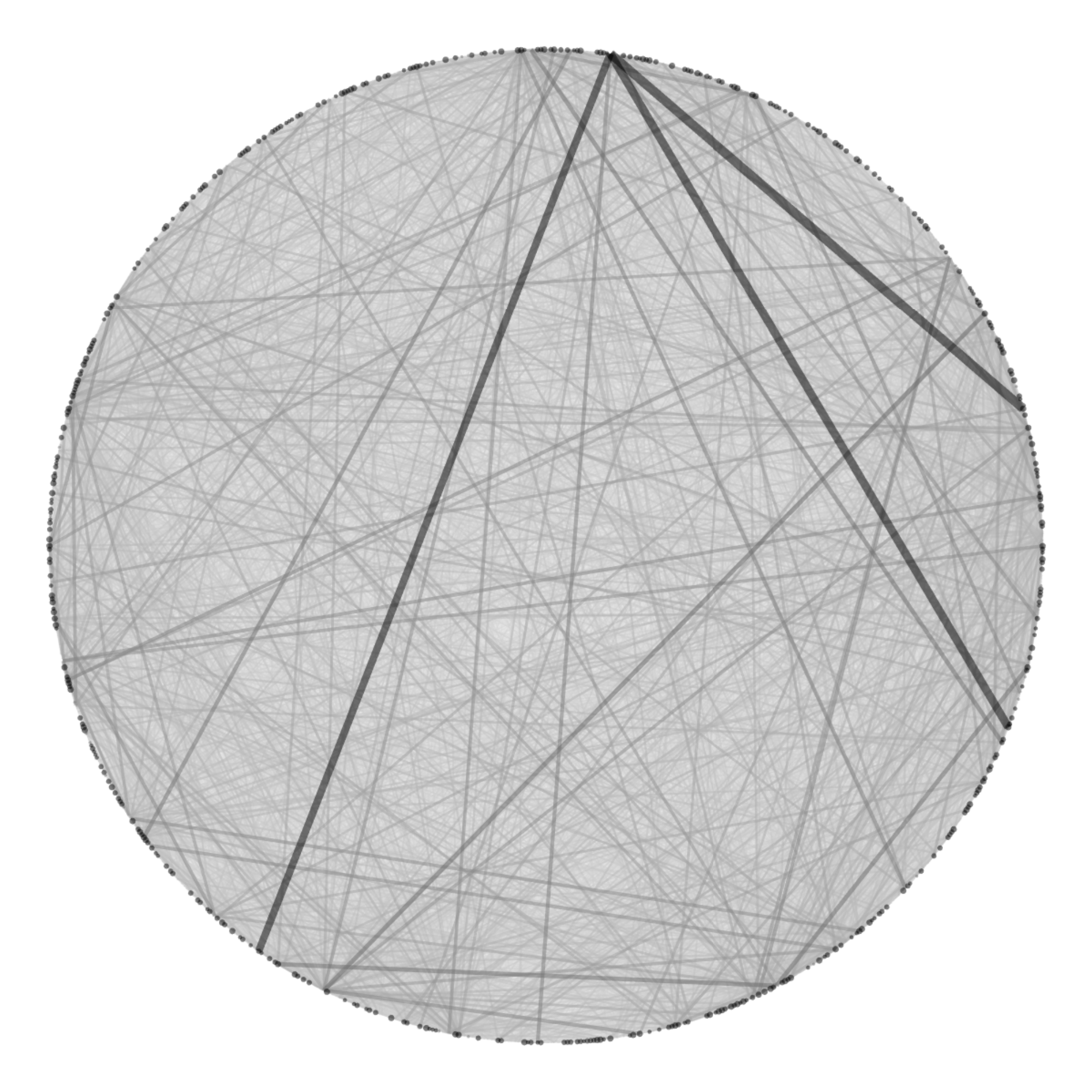}
		\caption{}
		\label{before}
	\end{subfigure}
	\hspace{6em}
	\begin{subfigure}{.35\textwidth}
		\centering
		\includegraphics[width=\textwidth]{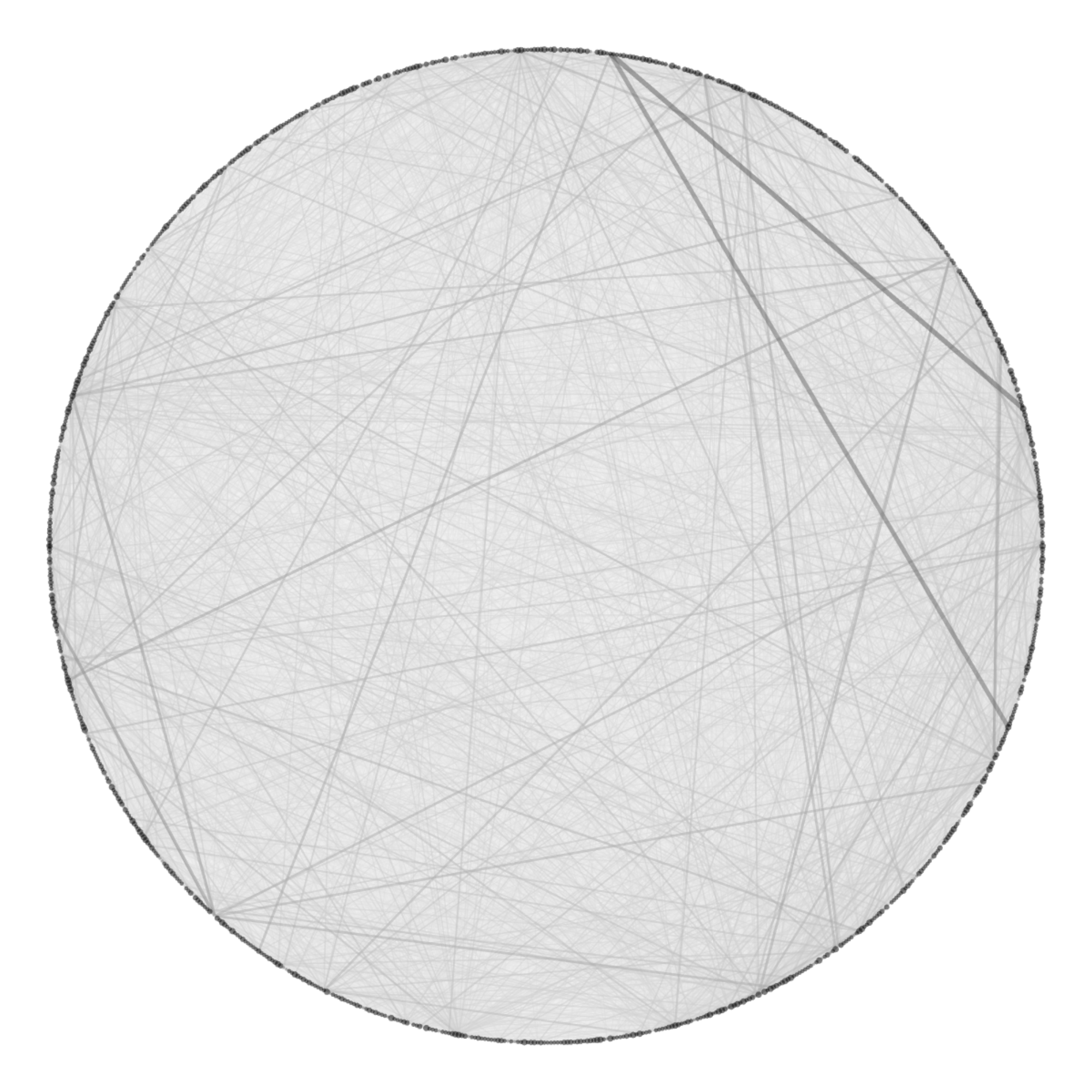}
		\caption{}
		\label{after}
	\end{subfigure}
	\caption{Undirected user similarity network built from 90-10 split of the MovieLens dataset with $U=943$ users and $K=40$. \subref{before} Existing ratings (1-5 scale) to the movie Toy Story. Nonexistent ratings were set to 0. \subref{after} Predicted ratings for the movie Toy Story. Node colors and radii are proportional to graph signal intensity. Edge colors and widths are proportional to signal difference between nodes.}
	\label{fig:user_net}
\end{figure*}	

In the user-based approach, each graph signal corresponds to a different movie and consists of the existing ratings by every user in the network who has rated that movie. A depiction of such a signal is shown on Figure \ref{before}. The way in which we create training and test samples from these signals is by ``zero-ing out'' the ratings of the user $u$ in which we are interested. The GNN is then trained to predict ratings by this user to any movie previously rated by other users. We can see this task as equivalent to completing the $u$th row of the rating matrix.

In the movie-based approach, the underlying graph is a movie similarity network. There are as many graph signals as users, and each of the signals correspond to the ratings that the corresponding user has given to the movies in the dataset. We create training and test samples by ``zero-ing out'' the ratings to a movie $m$ of our choice. The GNN is trained to predict ratings to this movie by any user who has already given ratings to other movies in the graph. This is equivalent to completing the $m$th column of the rating matrix. An example of this is given in Figure \ref{after}, where we predicted every rating in the 1st column of the matrix (corresponding to the movie Toy Story) and represented them as a graph signal on top of the user similarity network.

\blue{We choose five 90-10 splits for the training and test sets in both the user-based and movie-based experiments, submitting the GNNs to 40 epochs of training, in batches of 5 and without dropout. We optimize the cross entropy loss with learning rate 0.005. In both cases, we contrast the average performance of our localized activation function-based GNNs on all data splits with that of an all-ReLU GNN. On the best data split, we also compare our method with the recommender systems proposed in \cite{weiyu18-movie} and \cite{monti17-movie}.} In \cite{weiyu18-movie}, the authors also make the distinction between a user-based and a movie-based approach. The user-based approach predicts the entire rating matrix through application of linear graph filters defined on top of the user similarity network. They have up to 6 taps and their coefficients are optimized on the full 90,000-rating training set. This author's movie-based approach does the same, but using linear graph filters defined on the movie similarity network instead. These have up to 3 filter taps that are also optimized on the training set. Because our methods look at each user/movie individually, to make for a fair comparison we test the method in \cite{weiyu18-movie} once for each user/movie, taking only that particular user's/movie's ratings into account in the calculation of the RMSE. Additionally, \cite{weiyu18-movie} presents a third approach to the rating prediction problem: mirror filtering (MiFi), which filters on the user and movie similarity networks simultaneously. Although this is the best performing method in the analysis carried out in \cite{weiyu18-movie}, it cannot be compared to our approaches, because it intertwines user and movie information and does not allow looking at each user or movie individually.

As for the method in \cite{monti17-movie}, it uses a multi-graph CNN (MGCNN) to extract features from existing ratings. We train this CNN on the same 90,000-sample training dataset as before. The extracted features are then fed to a recurrent neural network (RNN) responsible for the score diffusion process. 

\begin{table}
	\centering
	\begin{tabular}{l|ccc|cc} \hline
		& \multicolumn{3}{c|}{\blue{Test RMSE}} & \multicolumn{2}{c}{\blue{Samples}} \\
		\blue{User}    & \blue{ReLU} & \blue{1h-max} & \blue{1h-med} & \blue{Train} & \blue{Test}	\\ \hline
		\blue{405}		& \blue{$1.4121$} & \blue{$\mathbf{1.3009}$} & \blue{$1.4135$} & \blue{$664$} & \blue{$73$}	\\
		\blue{655}		& \blue{$0.7809$} & \blue{$\mathbf{0.7384}$} & \blue{$0.7512$} & \blue{$617$} & \blue{$68$}	\\
		\blue{13}		& \blue{$1.3492$} & \blue{$1.4189$} & \blue{$\mathbf{1.3087}$} & \blue{$573$} & \blue{$63$}	\\
		\blue{450}		& \blue{$0.9441$} & \blue{$\mathbf{0.9121}$} & \blue{$0.9315$} & \blue{$486$} & \blue{$54$}	\\
		\blue{276}		& \blue{$0.7409$} & \blue{$0.7651$} & \blue{$\mathbf{0.7277}$} & \blue{$467$} & \blue{$51$}	\\ \hline
		\blue{5 users}		& \blue{$1.0454$} & \blue{$1.0271$} & \blue{$\mathbf{1.0265}$} & \blue{$2805$} & \blue{$311$}	\\ 
		\blue{10 users}		& \blue{$1.0261$} & \blue{$\mathbf{0.9895}$} & \blue{$1.0049$} & \blue{$4960$} & \blue{$551$}	\\ 
		\blue{15 users}		& \blue{$1.0171$} & \blue{$\mathbf{0.9841}$} & \blue{$0.9949$} & \blue{$6842$} & \blue{$760$}	\\ 
		\blue{20 users}		& \blue{$0.9867$} & \blue{$0.9684$} & \blue{$\mathbf{0.9621}$} & \blue{$8617$} & \blue{$957$}	\\ \hline
	\end{tabular}
	\caption{\blue{Average test RMSEs for ReLU, 1-hop max and 1-hop median by user, over 5 data splits. Number of samples in training and test sets.}}
	\label{table_user}
\end{table} 


\begin{table}
	\centering
	\begin{tabular}{lccc} \hline
		\blue{User}    & \blue{Local activation (min.)} & \blue{\cite{weiyu18-movie}} & \blue{\cite{monti17-movie}} \\ \hline
		\blue{405}		& \blue{$\mathbf{1.2681}$} & \blue{$1.4620$} & \blue{$1.2753$} \\
		\blue{655}		& \blue{$\mathbf{0.6138}$} & \blue{$0.8555$} & \blue{$0.7284$} \\
		\blue{13}		& \blue{$\mathbf{1.1659}$} & \blue{$1.5831$} & \blue{$1.3954$} \\
		\blue{450}		& \blue{$0.9230$} & \blue{$1.0341$} & \blue{$\mathbf{0.8986}$} \\
		\blue{276}		& \blue{$\mathbf{0.6504}$} & \blue{$1.1796$} & \blue{$0.9559$} \\ \hline
	\end{tabular}
	\caption{\blue{Comparison of test RMSEs obtained across different methods for best user data split ---localized activation GNNs, high order user-based linear graph filters \cite{weiyu18-movie}, multi-graph CNNs \cite{monti17-movie}.}}
	\label{table_user_comp}
\end{table}

In our user-based approach, the user similarity networks are built from the 90,000 ratings in the training set with $k=40$. This results in networks with average out-degree of 38.3. 
Using the same 90-10 \blue{training-to-test ratio} and $k=40$, we also build directed movie similarity networks, whose average out-degree is 1.09. Because these networks are large and highly connected at some nodes, only the 1-hop median and max were considered. The number of parameters in the convolutional layer of the ReLU, the 1-hop median and the 1-hop max architectures are shown in Table \ref{table_params}. The number of parameters of the 1-hop median/max only exceeds the number of parameters of the ReLU GNN by 2, because we regularized it by making $\bbw^1_\ell = \ldots = \bbw^F_\ell = \bbw_\ell$ the same for all features.

\blue{The 20 users with largest number of ratings were chosen to assess GNN performance in the user-based approach; we report test RMSEs for the first 5 and the averages for the first 5, 10, 15, and all 20 in Table \ref{table_user}. Localized activation functions consistently outperform the ReLU on average when the first 5, 10, 15 and all 20 users are considered, as well in most of the individual user cases (at least one of the local architectures outperforms the ReLU architecture for every user).} More than improving upon the ReLU, we note that our localized activation functions incur an increase in capacity given that they only have 2 more parameters than the ReLU GNN (cf. Table \ref{table_params}). On Table \ref{table_user_comp} we contrast the minimum RMSE achieved \blue{by either the max or median GNNs in the best data split --the data split where localized activations outperform the ReLU by the largest margin-- for each individual user} with the RMSEs obtained using the methods in \cite{weiyu18-movie} (user-based) and \cite{monti17-movie}. For \cite{weiyu18-movie}, we report the smallest RMSE of the 6 filters that are trained.
Our user-based method outperforms both \cite{weiyu18-movie} and \cite{monti17-movie} for all users except \blue{450}. Even then, the difference in the recorded RMSEs is minimal relatively to  discrepancies observed in other rows.

As for the movie-based approach, accuracies for \blue{the first 5 movies with most ratings, as well as averages for the first 5, 10, 15 and 20,} are shown in Table \ref{table_item}. Localized activation functions outperform the ReLU for all movies with as little as 2 additional trainable parameters (cf. Table \ref{table_params}), which once again attests to the increased capacity of median and max GNNs.  On Table \ref{table_item_comp}, \blue{the results obtained on the best data split for each user (the data split where localized activations outperform the ReLU by the largest margin)} are compared with those obtained using the movie-based method with 3 filter taps in \cite{weiyu18-movie} and the MGCNN in \cite{monti17-movie} \blue{on these same splits}. In both \cite{weiyu18-movie} and \cite{monti17-movie}, all 1,682 movies were taken into account. 

\blue{On Table \ref{table_item_comp}, localized activation function GNNs outperform \cite{weiyu18-movie} for all movies and \cite{monti17-movie} for 3 of the 5 movies. Except for ``Star Wars'' (where we outperform both methods with at least a 10\% reduction in RMSE) and ``Contact'' (where \cite{monti17-movie} outperforms our method by 5\%), for all other movies the differences in RMSE are not as significant as they were in the user-based approach. The most pertinent observation here is that, in the movie-base approach, our method is able to deliver recommendations that are essentially as accurate as those provided by \cite{weiyu18-movie} and \cite{monti17-movie}, but with less data and less computational complexity. Unlike \cite{weiyu18-movie} and \cite{monti17-movie}, in both the user and movie-based approaches we do not need to use the entire 90,000\blue{-sample} training dataset to train the GNN because we only look at a row/column of the rating matrix at a time. In this sense, another advantage of our architecture is the ability to offer more personalized and possibly on-demand movie recommendations.}


\begin{table}
	\centering
	\begin{tabular}{l|ccc|cc} \hline
		& \multicolumn{3}{c|}{\blue{Test RMSE}} & \multicolumn{2}{c}{\blue{Samples}} \\
		\blue{Movie}    & \blue{ReLU} & \blue{1h-max} & \blue{1h-med} & \blue{Train} & \blue{Test}	\\ \hline
		\blue{Star Wars}	& \blue{$0.9505$} & \blue{$0.9946$} & \blue{$\mathbf{0.9269}$} & \blue{$525$} & \blue{$58$}	\\
		\blue{Contact}	& \blue{$1.1337$} & \blue{$\mathbf{1.0836}$} & \blue{$1.0855$} & \blue{$459$} & \blue{$50$}	\\
		\blue{Fargo}	& \blue{$1.0411$} & \blue{$1.0994$} & \blue{$\mathbf{1.0403}$} & \blue{$458$} & \blue{$50$}	\\
		\blue{Return of the Jedi} & \blue{$0.9294$} & \blue{$\mathbf{0.9236}$} & \blue{$0.9577$} & \blue{$457$} & \blue{$50$}	\\
		\blue{Liar Liar} & \blue{$1.1926$} & \blue{$\mathbf{1.1908}$} & \blue{$1.1988$} & \blue{$437$} & \blue{$48$}	\\  \hline
		\blue{5 movies} & \blue{$1.0494$} & \blue{$1.0584$} & \blue{$\mathbf{1.0418}$} & \blue{$2333$} & \blue{$259$}	\\
		\blue{10 movies} & \blue{$1.1045$} & \blue{$\mathbf{1.0963}$} & \blue{$1.0974$} & \blue{$4377$} & \blue{$486$}	\\
		\blue{15 movies} & \blue{$1.0800$} & \blue{$\mathbf{1.0712}$} & \blue{$1.0793$} & \blue{$6185$} & \blue{$687$}	\\
		\blue{20 movies} & \blue{$1.0615$} & \blue{$\mathbf{1.0520}$} & \blue{$1.0580$} & \blue{$7845$} & \blue{$871$}	\\ \hline
	\end{tabular}
	\caption{\blue{Average test RMSEs for ReLU, 1-hop max and 1-hop median by movie, over 5 data splits. Number of samples in training and test sets.}}
	\label{table_item}
\end{table}


\begin{table}
	\centering
	\begin{tabular}{lcccc} \hline
		\blue{Movie}   & \blue{Local activation (min.)} & \blue{\cite{weiyu18-movie}} & \blue{\cite{monti17-movie}} \\ \hline
		\blue{Star Wars}		& \blue{$\mathbf{0.6823}$} & \blue{$0.7690$} & \blue{$0.7462$} \\
		\blue{Contact}		& \blue{$1.0290$} & \blue{$1.0300$} & \blue{$\mathbf{0.9746}$} \\
		\blue{Fargo}		& \blue{$0.8518$} & \blue{$1.0684$} & \blue{$\mathbf{0.8420}$} \\
		\blue{Return of the Jedi}		& \blue{$\mathbf{0.8402}$} & \blue{$0.8564$} & \blue{$0.8550$} \\
		\blue{Liar Liar}	& \blue{$\mathbf{1.1693}$} & \blue{$1.1708$} & \blue{$1.1697$} \\ \hline
	\end{tabular}
	\caption{\blue{Comparison of test RMSEs obtained across different methods ---localized activation GNNs, high order movie-based linear graph filters \cite{weiyu18-movie}, multi-graph CNNs \cite{monti17-movie}.}}
	\label{table_item_comp}
\end{table}

\blue{

\subsection{Citation networks} \label{subsec_cora}  

To evaluate the performance of the localized activation functions in a node classification setting, we compare max and median GNNs with GNN architectures using only ReLU activations on the Cora dataset.
The Cora dataset consists of $N=2708$ scientific articles that pertain to $C=7$ different classes and make up the nodes of a citation network. Each article is described by a bag-of-words feature vector with $F_{\mbox{{\scriptsize in}}}=1433$ words. Given the articles' feature vectors, the objective is to predict to which class each article belongs.

In our setup, feature vectors are interpreted as multi-feature graph signals and, during training, validation and test, all of them are fed to the GNN models. The GNNs generate intermediate features for all nodes, which are then interpreted as individual samples and processed through a fully connected layer mapping each node's features to a class label between 1 and 7. In the training stage, we only predict labels for $140$ nodes; the validation and test nodes are a total of $300$ and $1000$ respectively. This data split is the same used in \cite{kipf17-classifgcnn}, and can be found at \texttt{http://github.com/tkipf/pygcn}. 

$1$-hop and $2$-hop max and median GNNs with $L=1, F_0 = F_{\mbox{{\scriptsize in}}}, F_1 = 16$, and $K_1=5$ were compared against two ReLU architectures, which we call $\mbox{ReLU}_1$ and $\mbox{ReLU}_2$.  $\mbox{ReLU}_1$ has the same hyperparameters as the localized activation function GNNs, while $\mbox{ReLU}_2$ has hyperparameters $L=4, F_0 = F_{\mbox{{\scriptsize in}}}, \{F_i\}_{i=1}^4 = 16$, and $\{K_i\}_{i=1}^4=2$. $\mbox{ReLU}_2$ was designed so as to provide for a fair comparison with \cite{kipf17-classifgcnn}. Even if the GNN architecture in \cite{kipf17-classifgcnn} only considers convolutional filters with $1$-hop diffusions ($K=2$), by setting the number of layers to $L=4$ we can force information exchanges at most $4$-hops away, which is equivalent to having $K=5$. 

All models were trained by optimizing the cross entropy loss, for a total of $150$ epochs and with learning rate $0.005$. The classification accuracy achieved by each architecture is presented in Table \ref{table_cora}. Regardless of the number of hops or of the type of nonlinearity, the localized activation functions outperform both ReLU architectures by a significant margin, attesting to the value of encoding the graph structure in the computation of nonlinearities to improve GNN capacity.

\begin{table}
	\centering
	\begin{tabular}{l|cc|cc|cc} \hline
	   & \multicolumn{2}{c|}{\blue{ReLU}} & \multicolumn{2}{c|}{\blue{Max}} & \multicolumn{2}{c}{\blue{Median}} \\ 
		\blue{Architecture} & \blue{$L=1$}		& \blue{$L=4$} & \blue{1h} & \blue{2h} & \blue{1h} & \blue{2h} \\ \hline
		\blue{Accuracy (\%)} & \blue{$72.4$} & \blue{$46.5$} & \blue{$\mathbf{80.5}$} & \blue{$77.7$} & \blue{$78.8$} & \blue{$78.4$} \\ \hline
	\end{tabular}
	\caption{\blue{Cora test classification accuracy for $\mbox{ReLU}_1$, $\mbox{ReLU}_2$ and $1$-hop and $2$-hop max and median GNNs.}}
	\label{table_cora}
\end{table}

}

%% file: conclusions-nonlinear.tex


We have presented GNN architectures that replace pointwise nonlinearities by activation functions with multiple local inputs. These activation functions perform a linear combination of signals observed at the output of nonlinear max or median operators in node neighborhoods of increasing resolution. By using either max or median filters, we achieved greater model capacity at the expense of only a slight increase in computational complexity, with the architectures still being linear in the number of nodes. As weighted linear combinations of nonlinear operators, they also endow GNNs with the ability to learn activation function parameters from data. We have additionally shown that the gradients of localized activation functions can be efficiently computed through backpropagation.

Median and max GNNs were compared with GNNs using only pointwise activation functions in 3 different problems, and we observed performance improvements across all of them. In source localization on synthetic graphs, localized activation functions improved GNN capacity and outperformed the traditional ReLU-based designs regardless of the number of hops and of the graph degree. In authorship attribution, the classification accuracy improved in 1.8\% for Emily Br\"{o}nte and 0.4\% for Jane Austen, approaching a little more than 98\% of texts correctly classified as having been written or not by the author. \blue{We have additionally} proposed a user/movie-oriented movie recommendation system that is fit for online implementation and that improved upon \blue{both the conventional GNN implementation and two comparable methods. Finally, on the Cora dataset localized activation GNNs were shown to improve performance upon conventional GNNs with pointwise activation functions in at least 7\%.}